\documentclass[12pt]{article}
\usepackage{hyperref}          
\usepackage{graphics}
\usepackage{epigraph}
\usepackage{epstopdf}
\usepackage{amsthm}
\usepackage{amsmath}
\usepackage{amssymb}
\usepackage{amstext}
\usepackage{amsfonts}
\usepackage{enumerate} 
\usepackage{mathrsfs}
 \usepackage{float}
\usepackage{fixmath}
\usepackage[normalem]{ulem}

\usepackage[dvipsnames]{xcolor} 
\hypersetup{
    colorlinks, 
    linkcolor={red!50!black}, 
    citecolor={blue!50!black},
    urlcolor={blue!80!black}  
}
\usepackage[longnamesfirst]{natbib}

\usepackage[displaymath, mathlines]{lineno}  
\usepackage{microtype}
\usepackage{mathabx}  
\usepackage{footnote}
\makesavenoteenv{tabular}
\makesavenoteenv{table}
\usepackage[onehalfspacing]{setspace}
\usepackage[shortlabels, inline]{enumitem}
\usepackage{multirow}
\usepackage[title]{appendix}

\usepackage{pstricks} 
\usepackage{pst-plot}
\usepackage{pst-plot}
\usepackage{pst-func}
\usepackage{pst-node}
\usepackage{pst-math}
\usepackage{setspace} 
\usepackage{pst-plot}
\usepackage{cancel}  
\usepackage{sectsty}
\usepackage{caption, subcaption}
\usepackage[colorinlistoftodos,textsize=tiny]{todonotes}
\usepackage{lipsum}

\begin{document}

\baselineskip=.22in\parindent=30pt

\newtheorem{tm}{Theorem}
\newtheorem{dfn}{Definition}
\newtheorem{lma}{Lemma}
\newtheorem{assu}{Assumption}
\newtheorem{prop}{Proposition}
\newtheorem{cro}{Corollary}
\newtheorem*{theorem*}{Theorem}
\newtheorem{example}{Example}
\newtheorem{observation}{Observation}
\newcommand{\exm}{\begin{example}}
\newcommand{\exmm}{\end{example}}
\newcommand{\obs}{\begin{observation}}
\newcommand{\obss}{\end{observation}}
\newcommand{\cor}{\begin{cro}}
\newcommand{\corr}{\end{cro}}
\newcommand{\exx}{\end{exa}}
\newtheorem{remak}{Remark}
\newcommand{\rmk}{\begin{remak}}
\newcommand{\rmkk}{\end{remak}}
\newcommand{\thm}{\begin{tm}}
\newcommand{\nt}{\noindent}
\newcommand{\thmm}{\end{tm}}
\newcommand{\lm}{\begin{lma}}
\newcommand{\lmm}{\end{lma}}
\newcommand{\ass}{\begin{assu}}
\newcommand{\asss}{\end{assu}}
\newcommand{\df}{\begin{dfn}  }
\newcommand{\dff}{\end{dfn}}
\newcommand{\prp}{\begin{prop}}
\newcommand{\prpp}{\end{prop}}
\newcommand{\bqu}{\sloppy \small \begin{quote}}
\newcommand{\equ}{\end{quote} \sloppy \large}
\newcommand\cites[1]{\citeauthor{#1}'s\ (\citeyear{#1})}

\newcommand{\eq}{\begin{equation}}
\newcommand{\eqq}{\end{equation}}
\newtheorem{claim}{\it Claim}
\newcommand{\cl}{\begin{claim}}
\newcommand{\cll}{\end{claim}}
\newcommand{\eit}{\end{itemize}}
\newcommand{\ben}{\begin{enumerate}}
\newcommand{\een}{\end{enumerate}}
\newcommand{\bcen}{\begin{center}}
\newcommand{\ecen}{\end{center}}
\newcommand{\fn}{\footnote}
\newcommand{\ds}{\begin{description}}
\newcommand{\dss}{\end{description}}
\newcommand{\prf}{\begin{proof}}
\newcommand{\prff}{\end{proof}}
\newcommand{\cs}{\begin{cases}}
\newcommand{\css}{\end{cases}}
\newcommand{\lb}{\label}
\newcommand{\ra}{\rightarrow}
\newcommand{\tra}{\twoheadrightarrow}
\newcommand*{\supp}{\operatornamewithlimits{sup}\limits}
\newcommand*{\inff}{\operatornamewithlimits{inf}\limits}
\newcommand{\nf}{\normalfont}
\renewcommand{\Re}{\mathbb{R}}
\newcommand*{\mmax}{\operatornamewithlimits{max}\limits}
\newcommand*{\mmin}{\operatornamewithlimits{min}\limits}
\newcommand*{\argmax}{\operatornamewithlimits{arg max}\limits}
\newcommand*{\argmin}{\operatornamewithlimits{arg min}\limits}
\newcommand{\uhr}{\!\! \upharpoonright  \!\! }

\newcommand{\CR}{\mathcal R}
\newcommand{\CC}{\mathcal C}
\newcommand{\CT}{\mathcal T}
\newcommand{\CS}{\mathcal S}
\newcommand{\CM}{\mathcal M}
\newcommand{\CL}{\mathcal L}
\newcommand{\CP}{\mathcal P}
\newcommand{\CN}{\mathcal N}

\newtheorem{innercustomthm}{Theorem}
\newenvironment{customthm}[1]
  {\renewcommand\theinnercustomthm{#1}\innercustomthm}
  {\endinnercustomthm}
\newtheorem{einnercustomthm}{Extended Theorem}
\newenvironment{ecustomthm}[1]
  {\renewcommand\theeinnercustomthm{#1}\einnercustomthm}
  {\endeinnercustomthm}
  
  \newtheorem{innercustomcor}{Corollary}
\newenvironment{customcor}[1]
  {\renewcommand\theinnercustomcor{#1}\innercustomcor}
  {\endinnercustomcor}
\newtheorem{einnercustomcor}{Extended Theorem}
\newenvironment{ecustomcor}[1]
  {\renewcommand\theeinnercustomcor{#1}\einnercustomcor}
  {\endeinnercustomcor}
    \newtheorem{innercustomlm}{Lemma}
\newenvironment{customlm}[1]
  {\renewcommand\theinnercustomlm{#1}\innercustomlm}
  {\endinnercustomlm}

\renewcommand{\qed}{\tag*{$\blacksquare$}}

\subsubsectionfont{\normalfont\itshape}

\makeatletter
\newcommand{\customlabel}[2]{%
\protected@write \@auxout {}{\string \newlabel {#1}{{#2}{}}}}
\makeatother

\newtheoremstyle{mydefinitionstyle} 
    {\topsep}                    
    {\topsep}                    
    {}                   
    {}                           
    {\bf}                   
    {.}                          
    {.5em}                       
    {}  
\theoremstyle{mydefinitionstyle}
\newtheorem{defn}{Definition} 
\newtheorem*{defn*}{Definition} 
\newtheorem{exmp}{Example}
\newtheorem*{exmp*}{Example}
\newtheorem{hint}{Hint}
\newtheorem{rem}{Remark}

\def\qed{\hfill\vrule height4pt width4pt
depth0pt}
\def\reff #1\par{\noindent\hangindent =\parindent
\hangafter =1 #1\par}
\def\title #1{\begin{center}
{\Large {\bf #1}}
\end{center}}
\def\author #1{\begin{center} {\large #1}
\end{center}}
\def\date #1{\centerline {\large #1}}
\def\place #1{\begin{center}{\large #1}
\end{center}}

\def\date #1{\centerline {\large #1}}
\def\place #1{\begin{center}{\large #1}\end{center}}
\def\intr #1{\stackrel {\circ}{#1}}
\def\R{{\rm I\kern-1.7pt R}}
 \def\N{{\rm I}\hskip-.13em{\rm N}}
 \newcommand{\cprod}{\Pi_{i=1}^\ell}
\let\Large=\large
\let\large=\normalsize


\begin{titlepage}

\def\thefootnote{\fnsymbol{footnote}}
\vspace*{0.5in}

\title{ The Effects of Social Pressure on Fundamental Choices: \\ \vspace{0.5 em} Indecisiveness and Deferral\fn{
This project was initiated by a fruitful (and quite enjoyable) invitation  to Khan on April $6^{th}$ 2020 by the Italian authors.
A preliminary version of this paper was presented at the XVI annual AMASES (Association for Mathematics Applied to Social and Economic Sciences) 2021 meeting, and formed the subtext of his two consolidatory lectures on the ``Foundation of Rational Choice" given at the University of Catania in May 2023.
All authors are especially grateful to Kevin Reffett for his careful reading and knowledgeable encouragement; they also thank Bob Barbera, Davide Carpentiere, Ralph Chami, Ying Chen, Rosaria Distefano, Paolo Ghirardato, Ani Ghosh, Paola Manzini, Marco Mariotti, Lea Nicita, Joel Sobel, Larry Samuelson, Eddie Schlee, and Eric Schliesser for stimulating conversation and correspondence.
Khan acknowledges the support of 2022 JHU Provost's Discovery Award on \lq\lq Deception and Bad-Faith Communication"  towards this draft.
Angelo Petralia acknowledges the support of ``Ministero del Ministero dell'Istruzione, dell'Universit\`a e della Ricerca (MIUR), PE9 GRINS Spoke 8", project \textit{Growing, Resilient, INclusive, and Sustainable}, CUP E63C22002120006.
Alfio Giarlotta and Francesco Reito
acknowledge the support of \lq\lq PIA.CE.RI" (PIAno di inCEntivi per la RIcerca di Ateneo).
}}


\author{
Alfio   Giarlotta\fn{Department of Economics and Business, University of Catania, Italy. alfio.giarlotta@unict.it}
~ 
M. Ali Khan\fn{Department of Economics, Johns Hopkins University, Baltimore, MD 21218 USA. akhan@jhu.edu}  
~  
Angelo  Enrico Petralia\fn{Department of Economics and Business, University of Catania, Italy. angelo.petralia@unict.it}
~ 
Francesco Reito\fn{Department of Economics and Business, University of Catania, Italy. reito@unict.it} 
}

\vskip 1.00em

\vskip 1.00em

\vskip 1.00em

\baselineskip=.18in

\noindent {\bf Abstract:}
In mainstream neoclassical economics, utility maximization is the only engine
of individual action, and the {\it other} or the {\it social,} if it is modeled for decisions deemed fundamental, it is done as a tacit  externality parameter affecting an agent's  maximized payoff. 
And even when  hitched to a social reference point, a fully decisive and immediate response is invariably assumed. In this paper, we propose a non-standard articulation of the trade-off between
personal utility and social distance, one motivated by experimental evidence from psychology, management science, and economics. Our approach deconstructs non-recurrent consumer choice to two stages:
a non-decisive first stage in which a binary relation, called {\it one-many} ordering, yields an interval,
the {\it consideration set,} to which the deferred choice is confined; a decisive second stage in which
the distance from the average social  choice, and future social expectations,  are taken into
account in present utility. Finally, we embed this indecisive consumer in an exploratory  game-theoretic
setting, and show that indecisiveness and choice deferral may cause social loss.    {\hfill(169 words)}

 \vskip 0.50in


\medskip




\medskip


\end{titlepage}

\setcounter{page}{1}

\setcounter{footnote}{0}


\setlength{\abovedisplayskip}{-10pt}
\setlength{\belowdisplayskip}{-10pt}


\newif\ifall
\alltrue 

\smallskip


\tableofcontents

\medskip
\begin{center}
------------------------------------------------------------------------------------------------------------------------------
\end{center}
\pagebreak
\smallskip 
\bqu {\it Only a being-in-common can make possible a being-separated -- of a bond that unbinds by bonding}.
\hfill{Jean-Luc Nancy (1991)}\fn{See the Preface of Jean-Luc Nancy, {\it The Inoperative Community}. Minneapolis: University of Minnesota Press.} \equ

\section{Introduction} 
In this paper we offer  a model of choice that goes beyond the one in mainstream neoclassical economic theory to take both individual  desires  and social pressures into account.
Our approach appeals to classical behavioral criteria that involve both the  maximization of  personal utility and the  minimization of some measure of divergence between individual and social choice.
This is to say, that the model takes a social preference explicitly into account in a formulation of a two-stage decision. 
In the first stage, we consider an  agent, who is not able to determine a trade-off between these two competing criteria, and is thereby modeled as indecisive.  We formulate this indecisiveness through  a binary relation, which we call \textit{one-many 
ordering}: \lq`one'' because individual preferences are relevant; `\lq many'' because the pressure exerted by the surrounding environment  plays a role; 
\lq`ordering'' because the output is a reflexive transitive relation.
Thus, in this stage, a non-decisive agent discards all the alternatives that are dominated according to individual utility and social cost of 
once-in-a-lifetime choices, resulting in what we can now see as his \textit{consideration set}.
In the second stage the agent makes his {fundamental} choice from this consideration set previously determined, taking into account personal utility, present social distance, and expected future social distance.\fn{Once this draft of the paper was completed, the authors learnt of the work of \cite{JindaniYoung2025} that also develops  the trade-off between individual choices impacted by social norms, but rather than a focus on the choice, their interesting emphasis is on the evolution of the norm itself, to develop a framework to explain, for example, duelling in Europe, female genital mutilation in North Africa, and footbinding in China.
Our goals are a lot more modest.
We hope to engage this original contribution in future work. }

With this brief, almost telegraphic, outline of what we attempt in this paper,  we devote the rest of this introduction to framing our results in the context of the antecedent literature spanning a variety of disciplinary registers.
To begin with a paper chosen out of a rich stream of literature developed in the last two decades,
   we consider its seminal investigation of choice functions as a {\it rational shortlist method (RSM)}, one that elaborates  on the relevance of the psychological and marketing literatures.\fn{Our choice of a 2007 paper by Manzini and Mariotti may appear questionable and arbitrary if viewed either from a chronological perspective, or as our attempt to give salience and prominence to it over the rest of a multifaceted literature: our choice is simply based on the fact that it is most suited to connect our model to the antecedent literature that transcends the economic one.}   \cite{ManziniMariotti2007}  
write: 

\bqu  In line with some prominent psychology and marketing studies, in our model we assume that the decision maker uses sequentially two rationales to discriminate among the available alternatives. These rationales are applied in a fixed order, independently of the choice set, to remove inferior alternatives. This procedure \lq\lq sequentially rationalizes" a choice function if, for any feasible set, the process identifies the unique alternative specified by the choice function.\fn{The authors continue:\lq \lq In this case, we say that a choice function is a Rational Shortlist Method (RSM). Notable in this respect are the \lq\lq Elimination by Aspects" procedure of \cite{Tversky1972}  and the idea of \lq\lq fast and frugal heuristics" of  \cite{GigerenzerTodd1999}.
Similarly, this type of model is widely used and documented in the management/marketing literature. \cite{YeeDahanHauserOrlin2007} provide recent and compelling evidence of the use by consumers of \lq\lq two-stage consideration and choice" decision-making procedures, and also refer to firms taking account of this fact in product development.''
To re-emphasize, it is not our intention to reinterpret the \lq\lq shortlisting" idea in the context of developing preferences that include private consumption utilities and social distance concerns, but to draw the relevance to the literature of the marketing/management community.} \equ

\noindent It is this in two-stage analysis of a decision-making procedure that the management/marketing literature links up with our work, with the substantive difference that the first-stage in the results reported here is given a societal underpinning rather than an individual's own memory, attention or environmental influences.\fn{Staying with  Manzini and Mariotti, they follow up their 2007 work by further consideration of what they call \lq\lq two-stage threshold representations" and in the context of stochastic rather than deterministic choice, by a comprehensive investigation of so-called \lq\lq consideration sets"; see \cite{ManziniMariottiTyson2013} and \cite{ManziniMariotti2014} respectively.  Also see \cite{ManziniMariotti2012} and  Section 4 in \cite{Yildiz2016} for a comparison of $RSMs$ with the author's own \lq\lq list-rationalizable method." Also see Remark 1 below. \label{fn:mm} }

 This is a rich literature that has furnished a new direction to the theory of deterministic and stochastic choice: the former dating to Hicks, Allen  Samuelson and Georgescu-Roegen, and the latter to Block-Marschak, Luce, Richter and McFadden. It is fair to say that the pioneers would not recognize the shape that the theory has currently taken, and the empirical and experimental evidence that has accumulated and invoked to support it.\fn{See also \cite{RobertsLattin1991}, \cite{DesaiHoyer2000}, \cite{ChakravartiJaniszewski2003},	 \cite{GilbrideAllenby2004}, \cite{GigerenzerGaissmaier2011}, \cite{GigerenzerSturm2012}, \cite{SmithKrajbic2018}, and their references.
  Consideration sets came into economics relatively late: see \cite{ChiangChibNarasimhan1999}, \cite{EliazSpiegler2011},  \cite{MasatliogluNakajimaOzbay2012},\cite{CherepanovFeddersenSandroni2013},  and their references and followers. According to these models, the agent does not consider all available alternatives in his decision, but he selects the preferred item among those that have not been previously discarded.
  Note also that the elimination of dominated alternatives, and the associated formation of consideration sets, has been reported and examined also in consumer psychology, economics, and computer science in multi-attribute decision making \citep{HuberPaynePuto1982,Herne1999,XuandXia2012}, but in this work we use it to describe indecisiveness of agents when they make what for them are fundamental choices.}
  By now, rather than seeking newer  perspectives and interpretations, and yet different directions to develop  different bases for consideration sets and two-stage procedures, the need is more for consolidatory  syntheses of what has been accomplished.
In one such recent exploratory attempt, \cite{CarpentierePetralia2025} use consideration sets to develop an underlying rubric  of a variety of ambient models.\footnote{Consideration sets are referred as \lq\lq attention filters" in \cite{MasatliogluNakajimaOzbay2012}, and their terminology is utilized by Carpentiere and Petralia in their paper.
It bears emphasis that  the idea of  consideration sets needs to be distinguished  from \lq\lq random consideration" of objects of choice. 
The introduction of randomization  leads to conceptual complications that are being totally bracketed  here.}

With this brief overview of work in a sister discipline, we move on to the register of economics, beginning with the pioneering work of  \cite{Duesenberry1949} on the influence of ``other-regarding preferences'' in  important and non-recurrent decisions, fundamental choices so to speak, as schools  or career, savings and assets,  location etc., all impacted by the choices of a select peer group that the chooser is influenced by.\fn{See, for example,  \cite{HumlumKleinjansNielsen2012} and \cite{LentBrown2020}.  \label{fn:fund1}}
This concept has now morphed under the notion of \textit{reference dependence}, and has been widely investigated in economic theory in the last two decades.
 A recent review on the topic is provided by \cite{BursztynJensen2016}, who discuss how social distances and social image can be formalized within microeconomic models.
In particular, the following related literature can be singled out for the convenience of the reader.
 \cite{Akerlof1997} analyzes equilibria of different models of social distance in individual decisions, claiming that social benchmarks may cause inefficient allocations. 
\cite{Sobel2005} describes preferences and concerns of agents for status and reciprocity;
\cite{BenabouTirole2006} study the behavior of subjects who have preferences for both reputation and social esteem, showing that pro-social behavior can be ruled out; 
\cite{IjimaKamada2017} propose a network structure of agents, whose payoffs are affected by social distances.
The influence of the \lq\lq other'' in {fundamental choices} is more explicit in applied and experimental economics: 
\cite{CrosonandShang2008} report the effect of \lq\lq downward" social information in decisions to finance public goods;
\cite{GrinblattKeloharjuIkheimo2009} document a \lq\lq neighborhood  effect"  in  consumer purchases of automobiles;
\cite{EesleyandWang2017} experimentally verify the impact of social influence on career choice. 

The standard approach to analyze the effect of the reference dependence on individual choices considers a \textsl{decisive} agent.
Indeed, the agent is required to promptly determine a trade-off between (i) the distance from the reference point, and (ii) personal utility.
However, empirical and experimental evidence shows that people in high-pressure contexts do want to `kill two birds with one stone', trying to obtain the maximum level of personal utility, while being aligned with a social choice level.
In these cases, the agent may not be able to counterbalance individual preferences and social values, but that leads to some inevitable deferment.  An agent becomes indecisive and hesitantly introspective and thereby defers his choice.
This behavioral phenomenon has been already reported by \cite{TverskyShafir1992} who claim that choices under conflicts may yield deferred decisions, contradicting the value maximization principle. 
Later on, indecisiveness and choice procrastination under social pressure and tensions has been tested and measured by \cite{FrostShows1993}, \cite{MilgramTenne2000}, and  \cite{Rassinetal2007}, who analyze psychometrics properties of the so-called \textit{Indeciveness Scale}.\footnote{The \textit{indecisiveness scale} consists of 15 items that are answered on a 5-point scale, and it measures individual differences in indecisiveness caused by many factors, such as regret and social anxiety.}
In a series of experiments \cite{KrijnenZeelenbergBreugelmans2015} and \cite{AlisonShortland2019}  show that important decisions, which may have a considerable impact on the life of subjects and centre on their ability to select between colliding values, are often deferred.
  
{Closer to the basic thrust of this paper, there has been recent work by \cite{PejsachowiczToussaert2017}, \cite{Gerasimou2018}, and \cite{OkTserenjimid2002} that characterizes indecisiveness in decision theory, often by requiring that individual preferences are incomplete, and assuming that agents, in the very first stage of their decision, can only discard suboptimal alternatives.} 
These developments in economic theory have influenced our approach towards capturing an agent's hesitancy when he makes fundamental choices under social pressure.
{Indeed, in our setting indecisiveness arises from the incompleteness of a binary relation which ranks alternatives according to personal utility and current social distance. 
Options characterized by contrasting levels of these two attributes determine a consideration set, from which the agent will pick an item that maximizes his comprehensive utility, which takes into account also current and future social distances.
We prove the existence, {under some assumptions about the agent's primitives}, of an optimal choice performed by an indecisive agent suffering from social pressure.
The consequences of choice deferral in strategic interactions are formalized, by introducing a novel notion of equilibrium, which accounts for indecisiveness of agents.  
Finally, we show that indecisiveness and choice deferral may produce social losses, because the agent discards alternatives that turn out to be desirable when he takes the definitive decision.}

 The organization of the paper is as follows.
Section \ref{SECT:model} introduces preliminary notation and describes the model.
Sections \ref{SECT:results1} and \ref{SECT:results2} show the results: the first focuses on individual's problem, and the second extends our approach to a game theoretic framework. 
Section \ref{SECT:Concluding remarks} underscores the exploratory nature of the work, and suggests directions for further investigation. 
All proofs are relegated to the Appendix.

\section{Preliminary notation and set-up}\label{SECT:model}

There is an individual who lives for three time periods, $t_{1},t_{2}$ 
and $t_{3}$.
He has to make a fundamental choice that involves the first and second period, and produces effects in the third.
{Such} choice is sequential, and will result at the end of the second period in the variable denoted by $x$.
A reference group imposes social pressure on the individual, in the sense that the latter faces the cost $c_1(\delta_{x,x_{s}})$ of the distance, denoted by $\delta$, between $x$ and $x_{s}$, the (finite) average group choice.

In the first period, the individual makes a \lq myopic\rq$\,$ action, based on a {\it personal} utility function $u(x)$, and $c_1(\delta_{x,x_{s}})$.
{Indeed, he} takes into account the current social distance (producing effects at $t_1$ and $t_2$), but not  the future social distance he will face at $t_{3}$. 
Moreover, we assume that at this stage the individual can only restrict the set of considered alternatives, using a mechanism that will be explained below. 
At $t_2$ the choice will be definitive, since the agent selects in his consideration set an alternative, {by} taking into account (trough the maximization of a function $U$) personal utility,  the present social cost, and the cost of the  distance between $x$ and $\bar{x}_f$, which represents the expected social group's choice that will be prevalent in future time $(t_{3})$, after the fundamental choice is made. 

The agent's utility can be written as


\begin{equation*}
U\big(u(x),c_1(\delta_{x,x_{s}}),c_{\,2}(\delta_{x,\bar{x}_f})\big).
\end{equation*}

The arguments of the function are as follows.
\begin{itemize}
\item $u\colon\mathbb{R}_{+}\rightarrow \mathbb{R}$ is a function that attains its finite maximum at $x^{\star}$. 
This map accounts for personal benefits and costs of the choice $x$, with no social concerns.
\item The map $x_f$ is a random variable on $\mathbb{R}_{+}$, with finite support. 
We denote 
by $\rho(x_f)$ the density of $x_f$, if any.
Furthermore, we denote by $\mathscr{F}$ the family of all random variables with finite support on $\mathbb{R}_{+}$.
The function $x_f$ describes the agent's belief about the future social group choice, that will be formed at $t_3$.
We denote by $\bar{x}_f$ the expected value of $x_f$.

\item $\delta\colon \mathbb{R}_{+}^{2}\rightarrow \mathbb{R}_{+}$ is a metric.
We denote by $\delta_{x,x_s}$ the distance between any $x\in\mathbb{R}_{+}$ and $x_s\in \mathbb{R}_{+}$.
Moreover, we denote by $\delta_{x,{\bar{x}_f}}$ the  distance between any $x\in \mathbb{R}_{+}$ and $\bar{x}_f\in\mathbb{R}_+$.
The map $\delta$ measures the distance between any pair of individual, or social, actual, or expected, fundamental choices.  

\item The maps $c_1\colon\mathbb{R}_+\to\mathbb{R}_{+}$ and $c_2\colon\mathbb{R}_{+}\to\mathbb{R}_{+}$ are nondecreasing  functions of the distance $\delta$ of $x$ respectively from either $x_{s}$ or $x_{f}$, that is,  
$\partial c_1/\partial\delta_{x,x_s}\geq 0$, and $\partial c_2/\partial\delta_{x,\bar{x}_f}\geq 0$.
Moreover, we assume that $c_1(0)=c_2(0)=0$.
We denote by $c_1(\delta_{x,x_s})$ the value, according to $c_1$, of the distance between $x$ and $x_s$, and by $c_2(\delta_{x,\bar{x}_f})$ the value, according to $c_2$, of the distance between $x$ and $\bar{x}_f$.
The functions $c_1$ and $c_2$ describe the cost of current and expected future social distance the agent faces, once he makes his choice.
\item $U\colon \mathbb{R}_{+}^3\to\mathbb{R}$ is a function nondecreasing in $u$, and nonincreasing in $c_1$ and $c_2$, that is $\partial U/\partial u(x)\geq 0$, $\partial U/\partial c_1(\delta_{x,x_s})\leq 0$, and $\partial U/ \partial c_2(\delta_{x,\bar{x}_f})\leq 0$.
Note that, when $u,c_1,c_2,$ and $x_f$ are exogenous this map can be written as a composite function $U\colon
\mathbb{R}^2_+\to\mathbb{R}$ of the vector $(x,x_s)$. 
It summarizes individual preferences at $t_1$.

\end{itemize}

In the first period, the individual defers his decision, and reduces the set of considered alternatives to those that are Pareto-optimal with respect to (a) personal utility, and (b) the cost of the distance of his choice from the current social one.
This approach describes agent's indecisiveness and allows no compensation between alternative criteria.
In the second period, the agent overcomes his indecisiveness, and makes his fundamental choice, considering also personal beliefs about the future choice of the reference group.
{His} decision process is formally described below.

\begin{enumerate}
\item At $t_{1}$, the agent observes $x_{s}$, and restricts his choice to the subset of alternatives that are maximal with respect to the binary relation $\succcurlyeq$ on $\mathbb{R}_{+}$ defined by

\begin{equation*}
x_{i}\succcurlyeq x_{j}\quad \; \Longleftrightarrow\;\qquad u(x_{i})\geqslant u(x_{j})\; \; \text{and}\; \;c_1(\delta_{x_{i},x_{s}})\leqslant c_1(\delta_{x_{j},x_{s}})
\end{equation*}

for all $x_{i},x_{j}\in \mathbb{R}_{+}$.
We call $\succcurlyeq$ the \textit{one-many ordering}, since it Pareto-ranks all alternatives according to the conflict between personal utility $u(x)$ (`one') and cost of $\delta_{x,x_{s}}$(`many').
As usual, $\succ$ denotes the strict part of $\succcurlyeq$, i.e., $x\succ x^{\prime}$ if and only if $x\succcurlyeq x^{\prime}$ and $x^{\prime}\not\succcurlyeq x$.
Let $\max(\mathbb{R}_{+},\succcurlyeq)=\{x\in\mathbb{R}_{+}\,\vert\, y\succ x\;\text{for no $y\in \mathbb{R_{+}}$}\}$ be the set of undominated alternatives in $\mathbb{R}_{+}$ with respect to $\succcurlyeq$.
In our setting $\max(\mathbb{R}_{+},\succcurlyeq)$ is the agent's \textit{consideration set}, which is shaped by the conflict between personal utility and social distance.

 
\item At $t_{2}$, the agent chooses $x$ from the consideration set determined at $t_{1}$, by taking into account personal utility $u(x)$, the cost of the distance $\delta_{x,x_s}$ of her choice from the current social choice, and the cost of the distance $\delta_{x,\bar{x}_f}$ from the expected future social choice.
The maximization problem becomes

\begin{equation*}
\max_{x\, \in\, \max\left(\mathbb{R}_{+},\succcurlyeq\right)} U\big(u(x),c_1(\delta_{x,x_s}),c_{\,2}(\delta_{x,\bar{x}_f})\big).\label{EQ:Optimal choice utility}
\end{equation*}

\end{enumerate}

Figure \ref{fig:diagram 1} pictures  the steps of the decision process.
 

\begin{figure}[pth]
\begin{center}
\psset{yunit=3.2cm,xunit=6.5cm,algebraic}
\begin{pspicture}[showgrid=false](0,-0.5)(2.2,0.5)
\psaxes[Dy=0,Dx=2,ticks=none,labels=none,origin=none,yAxis=false,arrowsize=0.2]{->}(0,0)(0.14,0)(2.44,0)
\psxTick(0.14){\text{\begin{minipage}{0.35\linewidth}
\footnotesize
\begin{center}
Individual observes $x_{s}$, and discards alternatives dominated according to $u(\cdot)$ and $c_1(\delta_{\cdot,x_{s}})$
\end{center} 
\end{minipage}}}
\psxTick(1.25){\text{\begin{minipage}{0.40\linewidth}
\footnotesize
\begin{center}
Individual, relying on $x_f$, maximizes over his \textit{consideration set} the function \vspace{-0.25cm}$$U\left(u(\cdot),c_1(\delta_{\cdot,x_s}),c_{\,2}\left(\delta_{\cdot,\bar{x}_f}\right)\right)$$  
\end{center}
\end{minipage}}}
\psxTick(2.20){\text{\begin{minipage}{0.40\linewidth}
\footnotesize
\begin{center}
$x_f$ is realized
\end{center}
\end{minipage}}}
\rput(0.15,0.15){$t_{1}$}
\rput(1.25,0.15){$t_{2}$}
\psxTick(1.25){}
\rput(2.20,0.15){$t_3$}
\end{pspicture}
\end{center}
\caption{The agent's decision process}
\label{fig:diagram 1}
\end{figure}
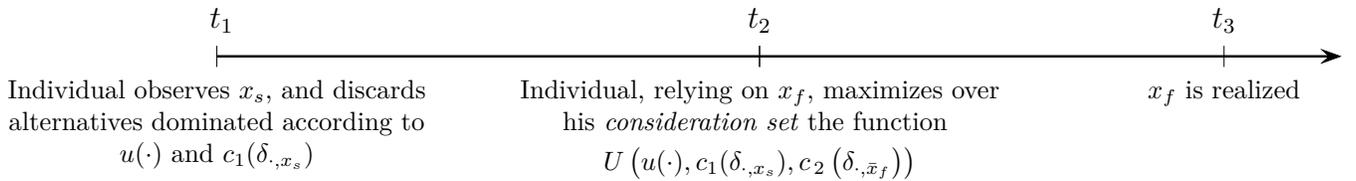

\vspace{1cm}
Note that the time frame between $t_1$ and $t_2$ is typically lower than the one starting at $t_2$ and finishing at $t_3$, since the consequences of a fundamental choice can be revealed after years (think again to education, or voting).\footnote{It can be observed that the timespan $t_{2}-t_1$ is a mere expositional device rather than a passage of time.
In fact, the time needed to decide between available alternatives may be considerable: think to purchasing a new house (after having discarded many apartments proposed by a real estate agency) or accepting a job offer (after having considered many potential placements).}   Thus, the social choice prevalent at $t_3$ may be, in principle, rather different from $x_s$.
This explains why $x_f$ is a random variable.
Note that the agent assumes his fundamental choice, performed at $t_2$, will be kept and produce effects at $t_3$.
This is distinctive of fundamental choices, which are usually thought as a \textsl{once-in-a-lifetime decisions} that are more or less irreversible over an individual's life cycle.\fn{We should like to see this  emphasis on choices of fundamentals, rather than \lq\lq choice" in general, as one underlying theme of this work, beginning with the title itself; also see Footnote \ref{fn:fund1} and  Footnote \ref{fn:fund2} below.}




%


\section{The individual's problem}\label{SECT:results1}  As already specified, we conceive the individual's  fundamental choice in two stages, a indecisive one and a decisive one.  
We begin with the first.

\subsection{Indecisiveness} \label{SUBSECT:non-compensatory approach}
As already mentioned, individuals may suffer from social pressure: this is the burden the agent feels when his choice is too far from that prevailing among the people around him.
After observing the distance of his personal preference from the reference point, the agent becomes indecisive, unable to fully resolve the contradiction between personal aspirations and the distance from the reference group's choice.
The best he can do at first is to defer his decision and reduce the set of considered alternatives, {according to two} distinct criteria:
(a) \textit{personal utility} of the fundamental choice, and (b) its  \textit{social distance}. 
Let us justify our approach by an enlightening example: the process to acquire an academic education. 
A society with a low level of education may regard high education useless, or even consider it with suspicious circumspection.  
Thus, the distance between the agent's desired education and the level of culture of the community may become a measure of social acceptability.
In this case, the conflict between the social level of choice and individual aspiration may not allow the agent to immediately take a decision.
This is true also in the opposite case in which a reference group may push the individual to pursue a high level of education.\fn{See, for instance, \cite{GuerraBraungartRieker1999} and \cite{SakaGatiKelly2008} who report difficulties of students college with regard to career choice caused by social pressure; also see  \cite{Zhangetal2022} and the references in Footnote \ref{fn:fund1} above. \label{fn:fund2} }  

The maximization of the one-many ordering captures the indecisiveness that arises from the conflict between personal aspirations and social pressure, and it constitutes the first stage of our approach.
The relationship $x_{i}\succcurlyeq x_{j}$ says that $x_{i}$ is at least as good as $x_{j}$ if both (1) $x_{i}$ gives at least the same personal utility as $x_{j}$, and (2) $x_{i}$ is at least as close to the reference point as $x_{j}$.
Thus, when the agent selects items which are maximal according to $\succcurlyeq$, he excludes all the  suboptimal alternatives that bring a lower personal utility and a higher social distance.
On the other hand, indecisiveness forces him to keep options that display unparalleled combinations of personal utility and social distance.
The relation $\succcurlyeq$ is a (typically incomplete) preorder, that is, reflexive and transitive.\footnote{Recall that $\succcurlyeq$ is \textit{complete} if $x\succcurlyeq y$ or $y \succcurlyeq x$ holds for all distinct $x,y \in X$, \textit{reflexive} if $x\succcurlyeq x$ holds for any $x\in X$, and \textit{transitive} if $x\succcurlyeq y$ and $y \succcurlyeq z$ implies  $x\succcurlyeq z$ for all $x,y,z\in X$.}
The next result shows that under some standard assumptions the set $\max(\mathbb{R}_{+},\succcurlyeq)$ is a closed interval.
\begin{prop}\label{PROP:One-many interval}
Assume $u$ is strictly quasiconcave, $c_1$ is increasing in $\delta$, and let $\delta$ be the euclidean distance on $\mathbb{R}_{+}$.
When $x_{s}=x^{\star}$, $\max(\mathbb{R}_{+},\succcurlyeq)=x_{s}=x^{\star}$.
Otherwise, $\max(\mathbb{R}_{+},\succcurlyeq)=
[\min\{x_{s},x^{\star}\},\max\{x_{s},x^{\star}\}]$.\footnote{A function $u\colon \mathbb{R}_{+}\to \mathbb{R}$ is \textsl{strictly quasiconcave} if for any $x,y\in \mathbb{R}_{+}$ and any $\lambda\in(0,1)$ we have that $u\left(\lambda x+(1-\lambda)y\right)>\min(u(x),u(y))$.
When the above inequality weakly holds for any $\lambda\in[0,1]$, we say  that $u$ is \textsl{quasiconcave}.}
\end{prop}

Thus, when individual preferences are convex, and the cost of present social distance is represented by an increasing map of the euclidean distance, three cases are possible:\footnote{Convexity of preferences has been tested by \cite{AndreoniCastilloPetrie2003}.
See also \cite{RichterRubinstein2019}, and \cite{BrandlBrandt2020} for recent advances on the topic.}

\begin{description}
\item[(i)] $x_{s}=x^{\ast}$:
the choice that maximizes individual utility is the social choice.
The maximum is unique, and is $x_{s}$.
The agent does not suffer from indecisiveness, since $x^{\star}$ is the unique alternative that maximizes personal utility and minimizes social distance.

\item[(ii)] $x_{s}<x^{\ast}$:
the social level of choice is lower than that maximizing individual utility.
The agent's consideration set is the interval is $[x_{s},x^{\ast}]$.
All distinct choices in the interval are pairwise incomparable, if the agent jointly evaluates their distance from the social choice and his personal utility.
Thus, the agent is not immediately able to take a choice from his consideration set: alternatives close to $x^{\star}$ bring him higher utility but also higher distance pressure, whereas choices in close proximity of $x_s$ exhibit less social distance but also lower personal utility. 

\item[(iii)] $x_{s}>x^{\ast}$:
the social choice is higher than that maximizing personal utility.
The agent's consideration set is $[x^{\ast},x_{s}]$.
This case is symmetric to the previous one. 
Again, the individual is indecisive, and unable to determine a combination of personal utility and social distance that satisfies his needs.
\end{description}

Note that the individual's commitment to $[\min\{x_{s},x^{\star}\},\max\{x_{s},x^{\star}\}]$ holds for two reasons: first, alternatives discarded at $t_{0}$ may not be available at $t_{1}$.
This happens in many real-life situations: once one receives a job offer, he must respect a deadline to sign the contract, otherwise the offer is revoked.
Moreover, considering and keeping all available alternatives at $t_{1}$, whenever possible, may be too costly for the agent.
 For instance, people who want to buy a new home usually have to pay a commission or deposit to guarantee their interest in the transaction.
We now turn to the analysis of the second stage. 


\subsection{Final choice}
\label{SUBSECTION:Choosing within the optimal interval}
Once his consideration set has been determined, the agent faces (in two cases out of three)  a new decision problem. 
In fact, he now takes as given a (not necessarily ex-ante) belief $x_{f}$ about the future social value that will show up at $t_3$, and chooses $x$ such that the tradeoff between personal utility, present social distance and future social distance is maximized.
Note that the agent does not consider his belief at time $t_1$ because of (at least) two reasons: (1) ``human  myopia", which leads him to first underweight the future consequences of his actions at $t_{3}$ and overweight the associated present consequences;\fn{See, for example,  \citealt{BrownLewis1981}, and \citealt{AngeletosHuo2021}.} (2) the possibility that he simply has no belief in the first stage, and shapes his expectations only when he overcomes the indecisiveness by a more accurate subsequent reasoning.\footnote{In this work for simplicity we assume that beliefs are exogenous.}
The next proposition offers sufficient conditions that guarantee the existence of a fundamental choice. 

\begin{prop}\label{PROP:Optimal choice}
\it If  $U$ is upper semicontinuous at $x$, $u$ is strictly quasiconcave, $c_1$ is increasing in $\delta$, and $\delta$ is the euclidean distance on $\mathbb{R}_{+}$, a {fundamental} choice exists.\footnote{A  function $U\colon \mathbb{R}_{+}\to\mathbb{R}$ is \textsl{upper semicontinuous} at $x_0\in\mathbb{R}_+$ if, for each $x\in\mathbb{R}_{+}$, we have that $\lim\sup_{x\to x_0} f(x)\leq f(x_0)$.
$U$ is upper semicontinuous if it is upper semicontinuous at any $x\in\mathbb{R}_{+}$.} 
\end{prop}

Thus, if we keep the assumptions of Proposition~\ref{PROP:One-many interval}, and we suppose that the individual comprehensive utility is upper semicontinuous, there exists a solution to the fundamental choice problem of an indecisive agent.\footnote{Upper semicontinuous utility functions, and continuous utilities, later applied in this work, {represent} 
complete preorders defined over $\mathbb{R}^{n},$ as shown by \cite{Richter1980}.
Thus, we are assuming that in the second stage the agent is rational.}
%
Note that a change in $x_{s}$, the social choice, affects the DM's consideration set $[\min\{x_{s},x^{\star}\},\max\{x_{s},x^{\star}\}]$, and the maximization of $U$ in the definitive choice. 
Similarly, personal utility $u$ and the cost $c_1$ of the current social distance play a role both in the formation of the agent's consideration set and in the final decision. 
Finally, the expected cost $c_2$ of distance from the expected future reference group's choice $\bar{x}_f$ only affects the agent's final choice.
We denote by $\widehat{x}=\argmax_{x\in\mathbb{R}_+} U\big(u(x),c_1(\delta_{x,x_s}),c_{\,2}(\delta_{x,\bar{x}_f})\big)$, if any.
Note that $\hat{x}$ is the solution to a standard choice problem, in which the agent is decisive, and able to immediately maximize $U$, considering present and future social distance.
Figure~\ref{fig:general model} illustrates the case in which $\widehat{x}$ belongs to the DM's consideration set $[\min\{x_{s},x^{\star}\},\max\{x_{s},x^{\star}\}]$: our predictions do not depart from neoclassical consumer theory.

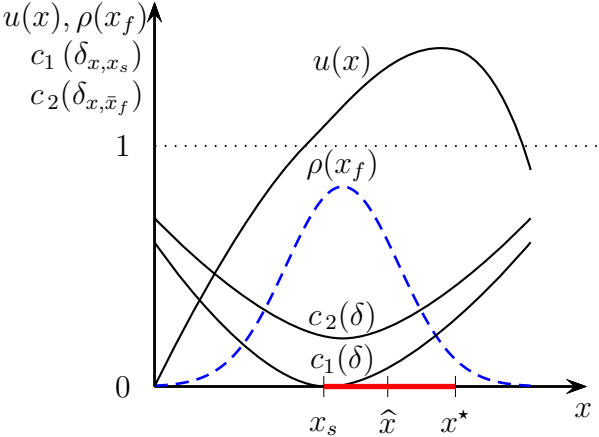
\begin{figure}[H]
\begin{center}
\psset{yunit=3.2cm,xunit=5cm,algebraic}
\begin{pspicture}[showgrid=false](0,-0.06)(1.2,2.1)
\psaxes[Dy=1,Dx=2,ticks=none,arrowsize=0.2]{->}(0,0)(0,0)(1.15,1.6)
\psline[linestyle=dotted](0,1)(1.2,1)
\psGauss[linestyle=dashed,linecolor=blue,mue=0.5,sigma=0.15,linewidth=1pt,yunit=1cm]{0}{1}
%
\pscurve(0,0.7)(0.50,0.2)(1,0.7)
\pscurve(0,0)(0.4,1)(0.8,1.4)(1,0.9)
\pscurve(0,0.6)(0.45,0)(1,0.6)
\rput(-0.21,1.53){$u(x),\rho(x_f)$}
\rput(-0.18,1.37){$c_1\left(\delta_{x,x_{s}}\right)$}
\rput(-0.18,1.20){$c_{\,2}(\delta_{x,\bar{x}_f})$}
\rput(1.14,-0.1){$x$} 
\rput(0.5,0.29){$c_{\,2}(\delta)$}
\rput(0.5,0.11){$c_1(\delta)$}
\psxTick(0.45){x_s}
\psxTick(0.62){\widehat{x}}
\psxTick(0.8){x^{\star}}
\rput(0.5,0.92){$\rho(x_f)$}
\rput(0.5,1.35){$u(x)$}
\psline[linecolor=red,linewidth=2pt](0.45,0)(0.8,0)
\end{pspicture}
\end{center}
\caption{The alternative $\widehat{x}$ falls in the optimality interval (the red segment).}
\label{fig:general model}
\end{figure}

However, the belief $x_f$ may drive $\widehat{x}$ out of $[\min\{x_{s},x^{\star}\},\max\{x_{s},x^{\star}\}]$, forcing the individual to choose an alternative with utility $U$ lower than the one he would have selected without social-pressure constraints, falling in a \textsl{indecisiveness trap}.
This case is depicted in Figure \ref{fig:trap}.
In Section \ref{SECT:results2} we will show how this condition in strategic interactions may generate a social loss.

\vspace{-1cm}

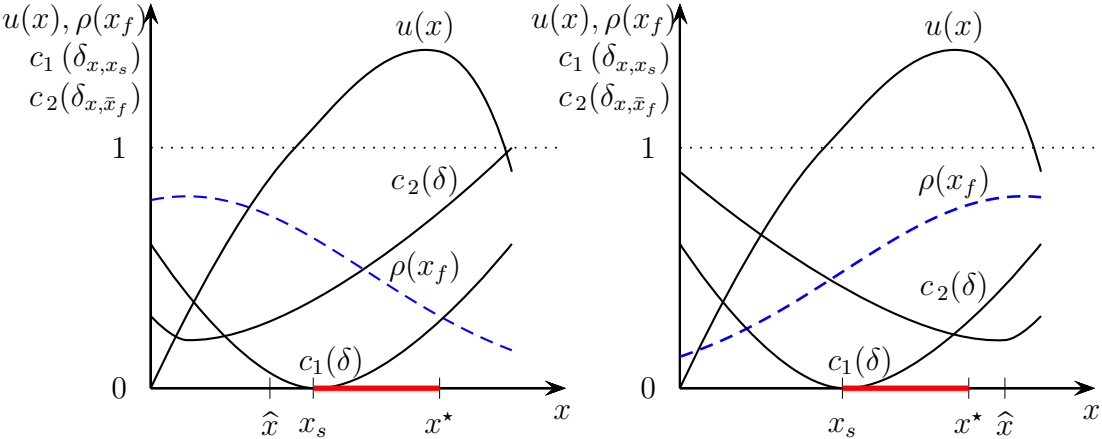
\begin{figure}[H]
\begin{center}
\psset{yunit=3.2cm,xunit=4.8cm,algebraic}
\begin{pspicture}[showgrid=false](-0.32,-0.06)(1.2,2.2)
\psaxes[Dy=1,Dx=2,ticks=none,arrowsize=0.2]{->}(0,0)(0,0)(1.15,1.6)
\psline[linestyle=dotted](0,1)(1.15,1)
\psGauss[linestyle=dashed,linecolor=blue,mue=0.1]{0}{1}
\pscurve(0,0)(0.4,1)(0.8,1.4)(1,0.9)
\pscurve(0,0.3)(0.1,0.2)(1,1)
\pscurve(0,0.6)(0.45,0)(1,0.6)
%
\rput(-0.21,1.53){$u(x),\rho(x_f)$}
\rput(-0.18,1.37){$c_1\left(\delta_{x,x_{s}}\right)$}
\rput(-0.18,1.20){$c_{\,2}(\delta_{x,\bar{x}_f})$}
\rput(1.14,-0.1){$x$}
\rput(0.76,0.86){$c_{\,2}(\delta)$}
\psxTick(0.45){x_s}
\psxTick(0.33){\widehat{x}}
\psxTick(0.8){x^{\star}}
\rput(0.76,0.5){$\rho(x_f)$}
\rput(0.76,1.50){$u(x)$}
\rput(0.5,0.11){$c_1(\delta)$}
\psline[linecolor=red,linewidth=2pt](0.45,0)(0.8,0)
\end{pspicture}
\hspace{1cm}
\begin{pspicture}[showgrid=false](0,-0.06)(1.2,2.2)
\psaxes[Dy=1,Dx=2,ticks=none,arrowsize=0.2]{->}(0,0)(0,0)(1.15,1.6)
\psline[linestyle=dotted](0,1)(1.2,1)
\psGauss[mue=0.95,linestyle=dashed,linecolor=blue,linewidth=1pt]{0}{1}
\pscurve(0,0)(0.4,1)(0.8,1.4)(1,0.9)
\pscurve(0,0.9)(0.9,0.2)(1,0.3)
\pscurve(0,0.6)(0.45,0)(1,0.6)
%

\rput(-0.21,1.53){$u(x),\rho(x_f)$}
\rput(-0.18,1.37){$c_1\left(\delta_{x,x_{s}}\right)$}
\rput(-0.18,1.20){$c_{\,2}(\delta_{x,\bar{x}_f})$}

\rput(1.14,-0.1){$x$}
\rput(0.76,0.86){$\rho(x_f)$}
\psxTick(0.45){x_s}
\psxTick(0.9){\widehat{x}}
\psxTick(0.8){x^{\star}}
\rput(0.76,0.42){$c_{\,2}(\delta)$}
\rput(0.76,1.50){$u(x)$}
\rput(0.5,0.11){$c_1(\delta)$}
\psline[linecolor=red,linewidth=2pt](0.45,0)(0.8,0)
\end{pspicture}
\end{center}
\caption{Extreme beliefs may move $\widehat{x}$ out of $[\min\{x_{s},x^{\star}\},\max\{x_{s},x^{\star}\}]$.}
\label{fig:trap}
\end{figure}

{\begin{rem}
As remarked in the Introduction, our approach to describe fundamental choices shares some features with the \textit{rational shortlist method} introduced by \cite{ManziniMariotti2007} to justify choice functions, i.e., experimental datasets that collect the subject's selections from many subsets of alternatives, called \textit{menus}.
In their work, the authors find the sufficient and necessary conditions under which choice functions can be interpreted as the decision of an experimental subject who selects, from each menu, the unique option that survives after the sequential application of asymmetric binary relations.
Later on, \cite{Garcia-SanzAlcantud2015} extend the framework of \cite{ManziniMariotti2007} to choice correspondences, which allow multiple selections from each menu. 
In our model choice is determined by the sequential maximization of the one-many ordering and the agent's comprehensive utility, but it does not have to be unique.
Thus, our approach can be represented as a specification of the pattern proposed by \cite{Garcia-SanzAlcantud2015}. 
To prove this fact, we need some additional notation, and definitions.
Let $X$ be a set of alternatives.
Recall that a binary relation $\succ\subseteq X^2$ is \textit{asymmetric} if $x\succ y$ implies $\neg(y\succ x),$ for any distinct $x,y\in X$.
We denote by $\mathscr{X}$ the family of all nonempty subsets of $X$, which are called \textit{menus}.
Given a menu $A\in\mathscr{X}$, we denote by $\max(A,\succ)$ the set $\{x\in A\,\vert\,y\succ x\,\text{for no}\,y\in A\}.$



\begin{defn}
Given $\mathscr{B}\subseteq \mathscr{X}$, a \textsl{partial choice correspondence} on $X$ is a map $\Gamma^* \colon \mathscr{B}\rightarrow \mathscr{X}$ such that $\Gamma^*(A)\subseteq A$ for any $A\in\mathscr{B}$. 	
\end{defn}
A partial choice correspondence displays the behavior of an experimental subject who selects some items from each menu. 
Partial choice correspondences can be explained by the sequential application of two asymmetric binary relations.
%


\begin{defn}[\citealt{Garcia-SanzAlcantud2015}]
Given  $\mathscr{B}\subseteq \mathscr{X}$, a partial choice correspondence $\Gamma^*\colon \mathscr{B}\to \mathscr{X}$ is \textit{rational by two sequential criteria}  if there is a pair of asymmetric binary relations $(\succ_1,\succ_2)$ such that $\Gamma^*(A)=\max(\max(A,\succ_1),\succ_2)$ for any $A\in \mathscr{B}.$ 
\end{defn}
We now define a map that incorporates the assumptions of our model.


\begin{defn}
Assume that $X=\mathbb{R_+},$ and $\mathscr{B}=\{\mathbb{R_+}\}$.
 Given $u\colon\mathbb{R}_{+}\rightarrow \mathbb{R}$, $x_s\in\mathbb{R}_{+}$, $x_f\in\mathscr{F}$,  $\delta\colon \mathbb{R}_{+}^{2}\rightarrow \mathbb{R}_{+}$,  $c_1\colon\mathbb{R}_+\to\mathbb{R}_{+}$, $c_2\colon\mathbb{R}_{+}\to\mathbb{R}_{+}$, and  $U\colon \mathbb{R}_{+}^3\to\mathbb{R}$, let $\Gamma^{\,\succcurlyeq}\colon\mathscr{B}\to \mathscr{X}$ be defined by $$\Gamma^{\,\succcurlyeq}(\mathbb{R}_+)=\max_{x\, \in\, \max\left(\mathbb{R_+},\succcurlyeq\right)} U\big(u(x),c_1(\delta_{x,x_s}),c_{\,2}(\delta_{x,\bar{x}_f})\big).$$

\end{defn}

As announced, we have:

\begin{prop}\label{PROP:two_sequential_rat}
If  $U$ is upper semicontinuous at $x$, $u$ is strictly quasiconcave, $c_1$ is increasing in $\delta$, and $\delta$ is the euclidean distance on $\mathbb{R}_{+},$ then $\Gamma^{\,\succcurlyeq}$ is a partial choice correspondence, and it is rational by two sequential criteria.  
\end{prop}
\end{rem}}




\section{Game-theoretic setting}\label{SECT:results2} 

In what follows, we introduce a game-theoretic setting in which {indecisive agents face some fundamental choice}, and behave according to the procedure described in Section~\ref{SECT:model}.
%

\subsection{The general case}\label{SUBSECT:the general case}

We now embed the individual in a society with $n\in \mathbb{N}_{+}$ agents. 
We denote by $x_{i}$ the fundamental choice of the $i$-th agent, whose utility is 


{\small
\begin{equation*}
U_i\left(u_i(x_{i}),c_{1}^{\,i}\left(\delta_{x_{i},g(x_{-i})}\right),c_{2}^{\,i}\left(\delta_{x_{i},\bar{h}_i\left(x^i_{f_{-i}}\right)}\right)\right).
\vspace{-0,1cm}
\end{equation*} }
 The arguments of the $i$-th agent's utility  function $U_i$ are game-theoretic elaborations of those described for the single-agent problem in Section~\ref{SECT:model}.

\begin{itemize}

\item $x_{i}\in \mathbb{R}_{+}$ is the $i$-th agent's fundamental choice performed at $t_1$.
We denote by  $x_{-i}\in \mathbb{R}_+^{n-1}$  the vector $(x_1,\dots,x_{i-1},x_{i+1},\dots,x_n)$ of choices made by the $n-1$ agents distinct from the $i$-th one.

\item $u_i:\mathbb{R}_{+}\to\mathbb{R}$ is a function that reaches a finite maximum at $x^{\star}_{i}\in\mathbb{R}_+$.
It represents $i$-th agent's personal utility.

\item $x^{i}_{f_{-i}}\in \mathscr{F}^{\,n-1}$ is a vector $\left(x^{i}_{f_1},\cdots,x^{i}_{f_{i-1}},x^{i}_{f_{i+1}},\cdots x^{i}_{f_{n}}\right)$ containing $n-1$ random variables with finite support on $\mathbb{R}_+$. 
 Each $x^{i}_{f_j}$ describes  the $i$-th agent's belief about the $j$-th agents's future choice at $t_3$.

\item $g_i\colon \mathbb{R}_{+}^{n-1}\to \mathbb{R}_+$ is a function that assigns to each vector $x_{-i}$ some value $g_i(x_{-i})$.
This map aggregates the choices of the rest of the society in a unique social reference point, that influences the $i$-th agent's decision.
The functional form of $g_i$ is not specified, since, in principle, the economic agent may aggregate the choices of the rest of the society in several ways.
However, an intuitive solution is to assume that $g_i(x_{-i})=\left(\sum_{j\neq i}x_{j}\right)/(n-1)$, the average choice of the $n-1$ players distinct from $i$.   

\item $h_i\colon \mathscr{F}^{\,n-1}\to \mathscr{F}$ is a function that associates to each vector $x^{i}_{f_{-i}}$ some random variable $h_i\left(x^{i}_{f_{-i}}\right)$ with finite support on $\mathbb{R}_+$.
We denote by $\bar{h}_i\left(x^{i}_{f_{-i}}\right)$ the expected value of $h_i\colon \mathscr{F}^{\,n-1}\to \mathscr{F}$.
The map $h_i$ assembles the $i$-th agent's beliefs about the future choice of any other individual of the society.
As for $g_i$, we do not identify $h_i$, but one may provide many meaningful specifications of it.
For instance, we can assume that $h_i$ is a random variable derived from a mixture of the probability distributions of the given beliefs, whose weights are determined by the importance that the $i$-th agent assigns to any other individual.      

\item $\delta\colon \mathbb{R}_+^{2}\to \mathbb{R}$ is a metric.
We denote by $\delta_{x_i,g_i(x_{-i})}$ the distance between any $x_i$ and each $g_{i}(x_{-i})$, and by $\delta_{x_i,\bar{h}_i\left(x^{i}_{f_{-i}}\right)}$ the distance between any $x_i$ and $\bar{h}_i\left(x^{i}_{f_{-i}}\right)$.
The map $\delta$ measures the distance of the $i$-th agent's choice at $t_2$ from the aggregated social choice at $t_2$, or from the expected aggregated social choice at $t_3$.

\item $c^{\,i}_{1}:\mathbb{R}_{+}\to \mathbb{R}_+$ 
 is a non-decreasing function of  $\delta_{x_{i},g(x_{-i})}$  
 such that $c^{\,i}_{1}(0)=0$ holds.
 This map accounts for the $i$-th agent's  social cost function, and measures the burden of the distances between  $x_i$ and and any other $j$-th agent's current fundamental choice.
\item $c^{\,i}_{2}:\mathbb{R}_{+}\to \mathbb{R}_+$ is a nondecreasing function of $\delta_{x_{i},\bar{h}\left(x_{f_{-i}}^i\right)}$,
 such that $c^{\,i}_{2}(0)=0$ holds.
This map describes the burden of the distances between  $x_i$ and and the expected aggregated fundamental choice $\bar{h}\left(x_{f_{-i}}^i\right)$, which will produce effects at $t_3$.

\item $U_i\colon \mathbb{R}_{+}^3\to \mathbb{R}$ is nondecreasing function of $u_i$, and it is nonincreasing with respect to $c_{1}^{\,i}$ and $c_{2}^{\,i}$.
It summarizes the $i$-th agent's preferences at $t_2$.
\end{itemize} 

Let $\succcurlyeq_{i,g(x_{-i})}$ be the binary relation on $\mathbb{R}_{+}$ defined by

\begin{equation*} \label{EQ:one-many Pareto}
 x_{i}\succcurlyeq_{i,g(x_{-i})} x_{i}^{\,\prime}\;\Longleftrightarrow\;u_{i}(x_{i})
 \geqslant	u_{i}\big(x_{i}^{\,\prime}\big)
  \;\; \text{ and } \;\; c_{1}^{\,i}\left(\delta_{x_{i},g(x_{-i})}\right)\leqslant 
  c_{1}^{\,i}\left(\delta_{x_{i}^{\,\prime},g(x_{-i})}\right),
\end{equation*}
for any $x_i,x_{i}^{\,\prime}\in\mathbb{R}_+$.
Let $\max(\mathbb{R}_{+},\succcurlyeq_{i,g(x_{-i})})=\{x\in\mathbb{R}_{+}\,\vert\, y\succ_{i,g(x_{-i})} x\;\text{for no $y\in \mathbb{R_{+}}$}\}$ be the set of maximal alternatives in $\mathbb{R}_{+}$ with respect to $\succcurlyeq_{i,g(x_{-i})}$.
As in Section~\ref{SECT:model}, $\max(\mathbb{R}_{+},\succcurlyeq_{i,g(x_{-i})})$ is the $i$-th agent's consideration set, that contains all the alternatives which bring incomparable combinations of personal utility and distance from the aggregated social choice.
In this strategic declination of the model described in Section~\ref{SECT:model}, the $i$-th agent's set of strategies is $\mathbb{R}_{+}$, and the set of strategy profiles is $\mathbb{R}_{+}^{n}$.
 Since the the functions $u_i,c_{1}^{\,i},c_{2}^{\,i},\delta,g_i,h_i$ and the beliefs $x^{i}_{f_{-i}}$ are exogenous in the problem, the $i$-agent's payoff function is 

$$U_i\left(u_i(x_{i}),c_{1}^{\,i}\left(\delta_{x_{i},g({x}_{-i})}\right),c_{2}^{\,i}\left(\delta_{x_{i},\bar{h}_i\left(x^i_{f_{-i}}\right)}\right)\right)=U_i(x_{i},x_{-i})=U_i\colon \mathbb{R}_{+}^{n}\to \mathbb{R}.$$
In what follows we introduce two different notions of equilibrium.

\begin{defn}\label{DEF:Equilibrium_standard}
 A vector $(x^{\diamond}_{1},\cdots,x^{\diamond}_{n})\in\mathbb{R}^{n}_{+}$ is an \textit{equilibrium} if, for any $i\in\lbrace 1,\cdots,n\rbrace$ and any $x^{\,\prime}_{i}\in\mathbb{R}_{+}$, $U_i(x^{\diamond}_{i},x^{\diamond}_{-i})\geq  
U_i(x^{\,\prime}_{i},x^{\diamond}_{-i})$  holds.

\end{defn}

\begin{defn}\label{DEF:Equilibrium_deferral}
 A vector $(x^{\diamond}_{1},\cdots,x^{\diamond}_{n})\in\mathbb{R}^{n}_{+}$ is an \textit{equilibrium after deferral} if, for any $i\in\lbrace 1,\cdots,n\rbrace$ and any $x^{\,\prime}_{i}\in\max(\mathbb{R}_{+},\succcurlyeq_{i,g(x_{-i})})$, $x^{\diamond}_{i}\in\max(\mathbb{R}_{+},\succcurlyeq_{i,g(x_{-i})})$ and    $U_i(x^{\diamond}_{i},x^{\diamond}_{-i})\geq 
U_i(x^{\,\prime}_{i},x^{\diamond}_{-i})$ hold.

\end{defn}
Definition~\ref{DEF:Equilibrium_standard} displays a standard equilibrium, typical of agents  who decide considering simultaneously personal utility, social distance, and expected future social distance of their fundamental choice.
Instead, Definition~\ref{DEF:Equilibrium_deferral} offers a strategic declination of the decision procedure described in Section~\ref{SECT:model}.
In an equilibrium after deferral, each $i$-th agent selects the alternative $x^{\diamond}_{i}$ that guarantees him the best payoff $U_i$, given the actions $x^{\diamond}_{-i}$ of the other agents, and her belief $x^{i}_{f_{-i}}$ about the aggregated future fundamental choice.
However, her choice must belong to $\max(\succcurlyeq_{i,g(x_{-i})},\mathbb{R}_{+})$.
This constraint encodes \textit{ex-post} the indecisiveness of the agent, who selects an alternative that maximizes his payoff function among those belonging to his consideration set, which is shaped by his initial inability to determine which option optimally compensates for personal aspirations and his need {for conformity}.
This framework allows to reproduce agents' indecisiveness and deferral in  fundamental choices, such as school and career decisions.
Note that an equilibrium after deferral is not necessarily an equilibrium: for some $i$-th agent the optimal choice, given the choice $x_{-i}$ of the rest of the society, may not belong to $\max(\succcurlyeq_{i,g(x_{-i})},\mathbb{R}_+)$.
However, given the constraint of Definition~\ref{DEF:Equilibrium_deferral}, in an equilibrium after deferral he is forced to select a strategy from his consideration set.
Similarly, there are equilibria that are not equilibria after deferral.
Indeed, in an equilibrium choices of agents do not necessarily fall into the consideration set determined by their one-many ordering.

To guarantee the existence of standard equilibria and equilibria after deferral, we impose some additional assumptions.
 First, we require that the fundamental choice of any agent is bounded above by a generic finite value $\dot{x}\in\mathbb{R}_+$.
 The parameter $\dot{x}$ is the highest level that any agent's choice can reach.
 For instance, in the choice of his high education, $\dot{x}$ can be interpreted as agent's decision to get a PhD.   
Similarly, if the fundamental choice is the purchase of a fancy new home, $\dot{x}$ may be a luxury apartment in the best area of the city. 
 We also assume that $U$ is continuous for each agent, and quasiconcave with respect to his fundamental choice.
 As expected, these two assumptions, and the boundedness of the choices of agents, ensure the existence of an equilibrium. 

 \begin{prop}\label{PROP:standard_equilibrium}
Assume that $x_{i}\leq \dot{x}$ for each $i\in\{1,\dots,n\}$, with $\dot{x}>0$.
 If  $U_i$ is continuous, and it is quasiconcave in $x_i$ for any $i\in\{1,\dots,n\}$, then there is an equilibrium.\footnote{
A function $U_i:\mathbb{R}_{+}^{n}\to{\mathbb R}$ is \textsl{continuous} 
if for any sequence $(\mathbold{x}_n)_{n\in\mathbb{N}}$ of vectors in $\mathbb{R}^n_{+}$ such that $\lim_{n\to\infty}(\mathbold{x}_n)=\mathbold{x}_{0}$, the equality $\lim_{n\to\infty}U_i(\mathbold{x}_n)=U_i(\mathbold{x}_{0})$ holds.}

\end{prop}

To guarantee the existence of an equilibrium after deferral, we have to rely on the behavioral conditions listed in Propositions \ref{PROP:One-many interval} and \ref{PROP:Optimal choice} that enable an individual to find an optimal fundamental choice, given the behavior of the rest of the society.

\begin{prop}\label{PROP:general_model}
Assume that $x_{i}\leq \dot{x}$ for each $i\in\{1,\dots,n\}$, with $\dot{x}>0$.
 If $U_i$ is continuous, and it is quasiconcave in $x_i$, $u_i$ is strictly quasiconcave, $c_{1}^{\,i}$ is increasing in $\delta$ for any $i\in\{1,\dots,n\}$, and $\delta$ is the euclidean distance on $\mathbb{R}_{+}$, then there is an equilibrium after deferral.

\end{prop}

 We may now enquire into the consequences of indecisiveness and choice deferral on the aggregated welfare.
Indeed, beliefs of agents about the expected fundamental choice of the rest of the society do not play any role in the formation of his consideration set, but they influence {his} ultimate decision.
As argued in Subsection~\ref{SUBSECTION:Choosing within the optimal interval}, it may happen that {an} agent's expectations about the future social fundamental choice would induce him to select an alternative that has been previously discarded, and does not belong to his consideration set.
In a strategic framework, indecisiveness may cause a social loss, which is measured by the gap between the aggregated welfare determined by a given equilibrium after deferral, and that of some more Pareto-efficient standard equilibrium.

\begin{defn}
Assume that $\mathbold{{x}^{\diamond}}=
(x^{\diamond}_1,\cdots,x^{\diamond}_n)\in\mathbb{R}^{n}_{+}$ is an equilibrium, but not an equilibrium after deferral, and ${\mathbold{x^{\diamond}}}^{\prime}=({x^{\diamond}_1}^{\prime},\cdots,{x^{\diamond}_n}^{\prime})\in\mathbb{R}^{n}_{+}$ is an equilibrium after deferral, but not an equilibrium.
	Assume also that  ${\mathbold{x^{\diamond}}}^{\prime}$ is Pareto dominated by $\mathbold{x^{\diamond}}$, that is, $ U_i(x^{\diamond}_i,x^{\diamond}_{-i})\geq U_i({x^{\diamond}_i}^{\prime},{x^{\diamond}_{-i}}^{\prime})$ holds for any $i\in\{1,\cdots,n\}$, and $ U_i(x^{\diamond}_i,x^{\diamond}_{-i})> U_i({x^{\diamond}_i}^{\prime},{x^{\diamond}_{-i}}^{\prime})$ at least one $i\in\{1,\cdots,n\}$.
	The \textit{deferral loss of ${\mathbold{x^{\diamond}}}^{\prime}$ with respect to $\mathbold{{x}^{\diamond}}$}  is $\mathscr{L}_{{\mathbold{x^{\diamond}}},{\mathbold{x^{\diamond}}}^{\prime}}=\sum_{i=1}^{n} U_i(x^{\diamond}_i,x^{\diamond}_{-i})-U_i({x^{\diamond}_i}^{\prime},{x^{\diamond}_{-i}}^{\prime})$.  
\end{defn}
The deferral loss measures the social cost of deferred {fundamental} choices.
When agents are indecisive, and delay their decision by confining the choice to their consideration set, they can incur into a drop of utility.
Indeed, their beliefs may lead them to prefer an alternative that is not in the consideration set, and it has been already discarded. 
Thus, each agent is forced to select, given the behavior of the rest of the society, the best item between those remaining in his consideration set.
 
If, \textit{ceteris paribus}, we drop the indecisiveness assumption, by assuming that agents promptly maximize their utility function, and do not defer their choice, more Pareto-efficient standard equilibria may arise.
Thus, $\mathscr{L}$ captures the difference between the aggregated welfare of a standard equilibrium and that of an equilibrium after deferral. 
As pointed out in the introduction, the cost of social pressure and tensions has been discussed in the antecedent literature\fn{See the references to 
	\cite{Akerlof1997}, \cite{AkerlofKranton2000}, \cite{BursztynJensen2016}, and \cite{UschevZenou2020}, among others.  } in which it is argued  that when  social concerns affect payoffs of agents, some non Pareto-efficient equilibria configurations occur.
However, these approaches do not take into account {indecisiveness of agents}, and they only investigate the gap of aggregated welfare between equilibria that arise in strategic interactions among decisive agents who maximize their comprehensive utility function.
Instead, the deferral loss, that is computed by comparing a standard equilibrium and an equilibrium after deferral,  quantifies the discrepancy of aggregated welfare that can be generated if we assume that agents are indecisive, due to social tensions determined by fundamental choices, and postpone their choice. 

In Propositions~\ref{PROP:standard_equilibrium} and~\ref{PROP:general_model} we displayed the conditions that allow the existence of standard equilibria and equilibria after deferral, but we do not characterize them.
Thus, in Subsection~\ref{SUBSECT:homogeneous_agents}, we assume that society consists only of two agents, sharing the same personal utility.
First, we show that any equilibrium after deferral is symmetric, and any symmetric equilibrium is an equilibrium after deferral.
Second, upon specifying beliefs, utility and cost functions of agents, we fully describe the collection of standard equilibria and equilibria after deferral, and we measure the deferral loss that arises from the indecisiveness of agents.

\subsection{Two agents}\label{SUBSECT:homogeneous_agents}

 We present a simplified framework, by assuming that society consists of only two representative individuals ($n=2$), so that each agent looks at the other as if he faces the rest of the society. 
Thus, the utility of each $i$-th agent is

\begin{equation*}
\label{EQ:main_1}
U_i\left(u_i(x_{i}),c_{1}^{\,i}\left(\delta_{x_{i},x_{-i}}\right),c_{2}^{\,i}\left(\delta_{x_{i},{\bar{x}}^{i}_{f_{-i}}}\right)\right).
\end{equation*} 
Since there are only two agents, $g_i(x_{-i})=x_{-i}$ and $h_{i}(x^{i}_{f-i})=x^{i}_{f-i}$ for any $i\in\{1,2\}$.
If economic agents defer their {fundamental} choice due to social tensions, and they share the same personal tastes, the equilibria after deferral exhibit the following features.

\begin{prop}\label{PROP:main proposition}
 Assume that $U_i$ is upper semicontinuous at $x_i$, $u_{i}$ is strictly quasiconcave, and $c_{1i}$ is increasing in $\delta$ for any $i\in\{1,2\}$. 
{ Moreover, suppose that $u_{1}(x)=u_{2}(x)$ for any $x\in\mathbb{R}_+$, and let $\delta$ be the euclidean distance on $\mathbb{R}_{+}$.}
If the pair $(x^{\diamond}_{1},x^{\diamond}_{2})$ is an equilibrium after deferral, then $x^{\diamond}_{1}=x^{\diamond}_{2}$.
If the pair $(x^{\diamond}_{1},x^{\diamond}_{2})$ is an equilibrium such that $x^{\diamond}_1=x^{\diamond}_2$, then $(x^{\diamond}_{1},x^{\diamond}_{2})$ is an equilibrium after deferral.
\end{prop}

In the setting described in Proposition \ref{PROP:main proposition}, individual preferences are aligned with the social ones, even if the two agents may be endowed with different costs, payoffs, and beliefs.
This situation is typical of social groups like political parties, religious associations, and sports clubs, whose affiliates are tied, at least partially, by common interests and aspirations.
Under the conditions that guarantee the existence of an optimal fundamental choice, in equilibria after deferral the two agents must take the same strategy.
Indeed, in an equilibrium after deferral, the first agent faces social pressure, and discards any alternative which is suboptimal under the joint evaluation of two criteria: his personal utility, maximized at $x^{\star}_{1}$, and the distance from the behavior of society, $x_{2}$.
Under the conditions of Proposition~\ref{PROP:main proposition}, he can select only the strategy that belongs to his consideration set, delimited by $x^{\star}_{1}$ and $x_{2}$.
However, the second agent (the rest of the society) acts according to the same procedure, he shares the same personal optimum $x^{\star}_{2}=x^{\star}_{1}=x^{\star}$, and he tends to perform his final choice in his consideration set, delimited by  $x^{\star}$ and  $x_{1}$.
As a consequence, in any stable configuration compatible with our behavioral procedure both agents are forced to select the same strategy.
Otherwise, should one of them select some alternative different from that picked by the opponent, his fundamental choice, or the opponent's one, would not belong anymore to the respective consideration set, hence contradicting Definition~\ref{DEF:Equilibrium_deferral}.

Moreover, any standard equilibrium in which agents adopt the same strategy is also an equilibrium after deferral.
Indeed, in {this case} each agent's choice, given the behavior of the other agent, maximizes his payoff, and, since the opponent's choice is equal, it belongs to his consideration set, {thus} satisfying the conditions of Definition~\ref{DEF:Equilibrium_deferral}.
The previous result provides a necessary condition of equilibria after deferral in a two-agent setting, and shows that standard symmetric equilibria are equilibria after deferral.
Depending on the assumptions about $u_i, c_{1i}, c_{2i}$, and $U_i$, it is possible to provide a full description the set of equilibria, equilibria after deferral, and the associated deferral loss, as showed in the following example.

\begin{exmp*}[\textsc{Conformist agents}]

\cite{Akerlof1997} presents a mathematical formalization of decisions in which social interactions matter.
In a society with two identical agents, each individual maximizes a utility function consisting of two components: (1) the intrinsic utility€™ of {his} choice, and (2) the social distance€ of his decision from the other individuals of the society.
One of the models discussed in Akerlof's work is the \textit{conformist} model.
A conformist agent maximizes
\vspace{-0.02cm}
$$U_i(x_i,x_{-i})=-d\vert x_{-i}-x_i\vert-ax_{i}^2+bx_{i}+k\,,$$
where $x_{-i}$ is the choice of the other agent, and $a$, $b$, $d$, and $k$ are nonnegative constants.
 The parameter $d$ measures the cost the that the $i$-th's agents faces when he does not adhere to the society's prevailing choice.
The equilibrium is each vector $(x^{\diamond}_1,x^{\diamond}_{2})$ such that $x^{\diamond}_1=x^{\diamond}_2=x^{\diamond}$ belong to $\left[(b-d)/2a,(b+d)/2a\right]$.
The author remarks that any equilibrium configuration in which each choice is different from $b/2a$, the individual optimum, is not Pareto efficient, and it brings a social loss caused by the need of conformism. 
Akerlof assumes that agents are immediately decisive, and capable to select the alternative which maximizes the difference between personal utility and social distance.
If we suppose that social pressure causes indecisiveness of agents and a deferral of their fundamental choices, and we keep the assumptions of the conformist model, that is, $u_i(x_i)=-ax_{i}^2+bx_{i}+k$, $c_{1i}\left(\delta_{x_i,x_{-i}}\right)=-d\vert x_{-i}-x_i\vert$, and $c_{2}^{\,i}=0$ (agents do not share any concern about future social distance), we find that the set of all equilibria after deferral is still the collection of vectors $(x^{\diamond}_1,x^{\diamond}_{2})$, such that $x^{\diamond}_1=x^{\diamond}_2=x^{\diamond}\in\left[(b-d)/2a,(b+d)/2a\right]$, and no deferral loss arises.\footnote{The proof is straightforward, and it is available upon request.}
 However, this result does not hold if agents are concerned in different ways about the expected future social cost of their decisions, and rely on extreme beliefs.
%
%
For instance, consider the game in which the payoff functions of the agents are 
$$ U_1(x_1,x_2)= -d\vert x_1-  \bar{x}_{f_{2}} \vert-d^{\,\prime}\vert x_{1}-x_2\vert-ax_{1}^2+bx_{1}+k\,,\;\;\text{and}$$
\vspace{-0.5cm}
$$ U_2(x_1,x_2)= -d\vert  x_2-\bar{x}_{f_{1}}\vert-d^{\,\prime\prime}\vert x_{2}-x_1\vert-ax_{2}^2+bx_{2}+k\,,$$
where  $x_{f_2}=x_{f_{1}}$ are random variables such that $P\left(x_{f_{2}}=10\,b/2a\right)=P\left(x_{f_{1}}=10\,b/2a\right)=1$, and $d^{\,\prime\prime}>d^{\,\prime}>d$.
Assume also that the choice of each agent is bounded above by some $\dot{x}\gg 10\,b/2a$.
First observe that, as in the conformist model, each $i$-th agent has personal utility $u_i=ax_{i}^2+bx_{i}+k$. 
Moreover, in this situation both agents have extreme beliefs about the future choice of the rest of the society, since they both believe that in the future the opponent will select $x=10\,b/2a$, a choice relatively greater than the individual optimum $b/2a$.
Finally, the condition $d^{\,\prime\prime}>d^{\,\prime}>d$ ensures that the second agent cares more about the distance of his choice from the future choice of his opponent than what the first agent does, and both agents are more concerned about the future consequences of their {fundamental} choice {than the current ones}.

This setting may lead to deferral losses.
This happens because in equilibrium after deferral each $i$-th agent, who has previously experienced indecisiveness caused by social distance, is forced to take a choice that lies between his personal optimum $b/2a$ and $x_{-i}$, the choice of the other agent.
However, his personal belief about the future behavior of the other agent, and the weight assigned to such expectation in his personal evaluation would have induced him to select, in a standard equilibrium, an alternative that does not belong to the interval previously determined.
To see this, assume that $a=2,b=d=4,d^{\,\prime}=7,d^{\,\prime\prime}=16$ and $k=5$.
The best response $\mathscr{B}_1(x_2)=\argmax_{x_1} U_1(x_1,x_2)$ of the first agent 
is 

$$
\mathscr{B}_1(x_2) =
\left\{ 
    \begin{array}{ll}
       7/4 &\;\text{if}\;x_{2}\in[0,7/4]\\
       x_2  &\;\text{if}\; x_2\in(7/4,15/4],\\
       15/4 &\;\text{if}\; x_2\in(15/4,\dot{x}].
       
    \end{array} 
\right.
$$
The best response $\mathscr{B}_2(x_1)=\argmax_{x_2} U_2(x_2,x_1)$ of the second agent 
is

$$
\mathscr{B}_2(x_1) =
\left\{ 
    \begin{array}{ll}
       4   &\;\text{if}\; x_1\in[0,4),\\
       x_1 &\;\text{if}\; x_1\in[4,6], \\
        6 & \;\text{if}\; x_1\in(6,\dot{x}].
    \end{array} 
\right.
$$
The unique equilibrium, according to Definition~\ref{DEF:Equilibrium_standard}, is the pair ${\mathbold{x}}^{\diamond}=(x^{\diamond}_1,x^{\diamond}_2)=(15/4,4)$.
However, Proposition~\ref{PROP:main proposition} holds, and it implies that $(15/4,4)$ is not an equilibrium after deferral.
The set of equilibria after deferral is
$$\{{\mathbold{x}^{\diamond}}^{\prime}=({x^{\diamond}_1}^{\prime},{x^{\diamond}_2}^{\prime})\,\vert\,{x^{\diamond}_1}^{\prime}={x^{\diamond}_{2}}^{\prime}={x^{\diamond}}^{\prime}\in[1,15/4]\}.\footnote{All the computations have been checked on \textit{Mathematica}, and are available upon request.} 
$$
The best response of each agent, the equilibrium, and the equilibria after deferral are represented in Figure~\ref{fig:best_responses}.

\begin{figure}[H]
\begin{center}
\psset{yunit=6.cm,xunit=6cm,algebraic}
\begin{pspicture}[showgrid=false](0,-0.06)(1.17,1.17)
\psaxes[Dy=6,Dx=6,ticks=none,arrowsize=0.2]{->}(0,0)(0,0)(1.15,1.15)
\rput(-0.22,1.10){$x_{2},\blue \mathscr{B}_{2}(x_1)$}
\rput(1.10,-0.09){$x_{1},\red \mathscr{B}_{1}(x_2)$}
\psyTick(0.6){$4$}
\psyTick(0.9){$6$}
\psxTick(0.2625){$7/4$}
\psxTick(0.5625){$15/4$}
\psline[linewidth=1pt,linecolor=blue](0.0,0.6)(0.6,0.6)
\psline[linewidth=1pt,linecolor=blue](0.6,0.6)(0.9,0.9)
\psline[linewidth=1pt,linecolor=blue](0.9,0.9)(1.12,0.9)
\psline[linewidth=1pt,linecolor=red](0.2625,0.0)(0.2625,0.2625)
\psline[linewidth=1pt,linecolor=red]{-]}(0.2625,0.2625)(0.5625,0.5625)
\psline[linewidth=1pt,linecolor=red](0.5625,0.5625)(0.5625,1.12)
\psline[linewidth=1pt,linestyle=dashed,linecolor=red]{[-}(0.14,0.15)(0.2625,0.2625)
\pscircle[fillcolor=blue,fillstyle=solid](0.5625,0.6){0.06}
\end{pspicture}
\end{center}
\caption{The best response functions respectively of the first agent (the red line), and of the second agent (the blue line).
The blue dot is the equilibrium, and the red segment delimited by the red brackets is the set of the equilibria after deferral.
}
\label{fig:best_responses}
\end{figure}
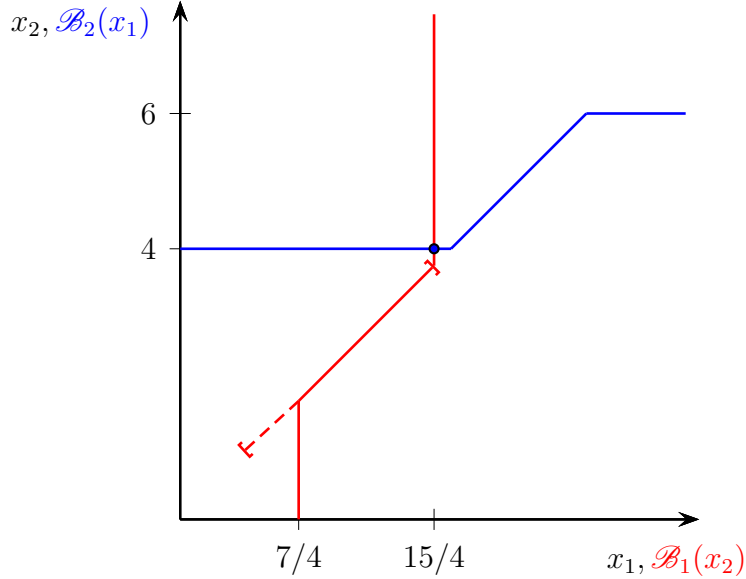

The pair $(15/4,4)$ is not an equilibrium after deferral, because the choice of the second agent does not belong to his consideration set $\max(\mathbb{R}_{+},\succcurlyeq_{2,x^{\diamond}_1})=[x^{\star}_{2},x^{\diamond}_1]=[1,15/4]$.
On the other hand, any pair $({x^{\diamond}_1}^{\prime},{x^{\diamond}_2}^{\prime})$ such that ${x^{\diamond}_1}^{\prime}={x^{\diamond}_2}^{\prime}={x^{\diamond}}^{\prime}\in[1,15/4]$ is an equilibrium after deferral.
To see why, note that the choice ${x^{\diamond}_1}^{\prime}$ of the first agent belongs to his consideration set $\max\left(\mathbb{R}_{+},\succcurlyeq_{1,{x^{\diamond}_2}^{\prime}}\right)=[x^{\star}_{1},{x^{\diamond}_2}^{\prime}]=[1,{x^{\diamond}_2}^{\prime}]$, and it maximizes $U_1$ in $[1,{x^{\diamond}_2}^{\prime}]$.
The choice ${x^{\diamond}_{2}}^{\prime}$ of the second agent belongs to his consideration set $\max\left(\mathbb{R}_{+},\succcurlyeq_{2,{x^{\diamond}_1}^{\prime}}\right)=[x^{\star}_{2},{x^{\diamond}_1}^{\prime}]=[1,{x^{\diamond}_1}^{\prime}]$, and it maximizes $U_2$ in $[1,{x^{\diamond}_1}^{\prime}]$.
However, note that the extreme belief $x_{f_2}$ of the first agent, and the weight $d^{\,\prime}$ assigned to it in his evaluation, moves his optimal response $\mathscr{B}_{1}({x^{\diamond}_2}^{\prime})$ far from $[1,{x^{\diamond}_2}^{\prime}]$.
Thus, if he had not deferred his choice and limited it to his consideration set, he would have selected $\mathscr{B}_{1}({x^{\diamond}_2}^{\prime})=7/4$.
Analogously, if the second agent had not experienced indecisiveness, he would have chosen $\mathscr{B}_{2}({x^{\diamond}_1}^{\prime})=4$.
As a consequence, some equilibria after deferral, when compared with the standard equilibrium point, cause a deferral loss.
For instance, the deferral loss associated with the equilibrium after deferral $(1,1)$, when compared with the pair $(15/4,4)$ is $\mathscr{L}_{{\mathbold{x}^{\diamond}},{\mathbold{x}^{\diamond}}^{\prime}}=\mathscr{L}_{(15/4,4),(1,1)}=U_1(15/4,4)-U_1(1,1)+U_2(15/4,4)-U_2(1,1)=32.125$.
Likewise, the deferral loss associated with the the equilibrium after deferral $(3/2,3/2)$, when compared with $(15/4,4)$ is $\mathscr{L}_{{\mathbold{x}^{\diamond}},{\mathbold{x}^{\diamond}}^{\prime}}=\mathscr{L}_{(15/4,4),(3/2,3/2)}=U_1(15/4,4)-U_1(3/2,3/2)+U_2(15/4,4)-U_2(3/2,3/2)=21.625$.

The deferral loss vanishes as the magnitude of the strategies of the considered equilibrium after deferral increases.
Indeed, when the strategies of an equilibrium after deferral are higher, and closer to the expected future choice $10b/2a=20$ of the rest of the society, the payoff of the first agent is higher than that received in the standard equilibrium configuration.
This happens because the loss that the first agent faces in the standard equilibrium from selecting an alternative $x^{\diamond}_1=15/4$ more distant from the personal optimum $b/2a=1$ and the choice $x^{\diamond}_2=4$ of the second agent exceeds now the gain obtained from selecting an alternative  closer to future social choice.
For instance, the equilibrium after deferral $({x^{\diamond}_1}^{\prime},{x^{\diamond}_1}^{\prime})=(2b/2a,2b/2a)=(2,2)$ does not Pareto dominates the equilibrium $(15/4,4)$ since $U_1(2,2)=-261>U_1(15/4)=-262.875$, and it does not bring any deferral loss, when compared with the standard equilibrium.
\end{exmp*}

\section{Concluding remarks}\label{SECT:Concluding remarks}

We provide a theoretical framework to capture and formalize {non-recurrent} choices with long term life-cycle consequences. 
Experimental and empirical literature in psychology and economics documents how agents constantly struggle with the contradiction between individual desires  and social pressures, leading to indecisiveness.
 At first, the best they can do is to discard dominated alternatives.
Then, they overcome indecision, and choose on the basis of the consideration set, personal preferences, and present and future social distance.
In strategic interaction indecisiveness may lead to social loss.
In fact, extreme (exogenous) beliefs may drive the individual optimum out of the consideration set previously determined, forcing the agent to select an item he would have not chosen if all alternatives were available.

Further research may be devoted to include a study of belief formation and updating mechanisms, and in particular to see how losses that arise from indecisiveness may be ruled out: see, {e.g.}, \cite{EpsteinNoorSandroni2007}, and \cite{AugenblickandRabin2021}.
Moreover, in the  game-theoretic setting that we presented a society consists of only  a finite number of agents.
In real-life situations, reference groups may be so large that the {fundamental} choice of just one individual is negligible and thereby does not affect the others.
Thus, a natural direction of research is to describe the effect of social distance in fundamental choices  using \textsl{large games}, in which players' names are distributed in a atomless probability space, and only a summary of societal actions play a role: see for example \cite{KhanetAl2013}.  

\section{Appendix: Proofs of results}\label{SECT:appendix} 

\noindent {\bf\textit{Proof of Proposition~\ref{PROP:One-many interval}}}. 
Since $u(\cdot)$ is strictly quasiconcave, it follows that $x^{\star}$ is unique and $u$ is increasing in $[0,x^{\star})$, and decreasing in $(x^{\star},+\infty ]$.
Consider the case $x_{s}=x^{\star}$.
Toward a contradiction, suppose there is $x^{\prime} \neq x_{s}$ belonging to $\max(\mathbb{R}^{+},\succcurlyeq)$.
We have $u(x^{\star})=u(x_{s})> u(x^{\prime})$ and $\delta_{x_{s},x_{s}}=0<\delta_{x^{\prime},x_{s}}$, which implies $c_1(\delta_{x_{s},x_{s}})=0 <c_1(\delta_{x^{\prime},x_{s}})$.
Thus $x_{s}\succ x^{\prime}$ and $x^{\prime}\not\in\max([0,+\infty),\succcurlyeq)$.

Next, consider the case $x_{s}<x^{\star}$.
We claim that $\max(\mathbb{R}_{+},\succcurlyeq)=[x_{s},x^{\star}]$.
To show that, we need  to prove that $(a)$ any point out of $[x_{s},x^{\star}]$ does not belong to $\max(\mathbb{R}_{+},\succcurlyeq)$, and $(b)$ any point in $[x_{s},x^{\star}]$ belongs to $\max(\mathbb{R}_{+},\succcurlyeq)$. 

To show that $(a)$ holds, assume by contradiction that there is $x^{\prime}\in\max(\mathbb{R}_{+},\succcurlyeq)$ that is not in $[x_{s},x^{\star}]$. 
Two cases are possible: $x^{\prime}<x_{s}$ or $x^{\prime}>x^{\star}$. 
If $x^{\prime}<x_{s}<x^{\star}$, (strict) monotonicity of $u$ implies $u(x_{s})\geq u(x^{\prime})$, and $\delta_{x_{s},x_{s}}=0<\delta_{x^{\prime},x_{s}}$, which implies $c_1(\delta_{x_{s},x_{s}}) < c_1(\delta_{x^{\prime},x_{s}})$, since $c_1$ is increasing.
Thus $x_{s}\succ x^{\prime}$ and $x^{'}\not\in \max(\mathbb{R}_{+},\succcurlyeq)$.
If $x^{\prime}>x^{\star}>x_{s}$, since $\delta$ is the euclidean distance on $\mathbb{R}_{+}$, we have that $\delta_{x^{\star},x_{s}}<\delta_{x^{\prime},x_{s}}$, which implies $c(\delta_{x^{\star},x_{s}})< c(\delta_{x^{\prime},x_{s}})$, and $u(x^{\star})\geq u(x^{\prime})$, since $x^{\star}$ is $\argmax_{x}u(x)$.
Thus $x^{\star}\succ x^{\prime}$, which implies that $x^{\prime}\not\in \max(\mathbb{R}_{+},\succcurlyeq)$.

To show that $(b)$ holds as well, take a point $x^{\prime}\in[x_{s},x^{\star}]$.
Toward a contradiction, suppose $x^{\prime}$ does not belong to $\max(\mathbb{R}_{+},\succcurlyeq)$.
This implies that there is a point $x^{\prime\prime}\neq x^{\prime}$ such that $x^{\prime\prime}\succ x^{\prime}$.
Without loss of generality, two cases are possible: $x^{\prime\prime}<x^{\prime}$ or $x^{\prime\prime}>x^{\prime}$.
If $x^{\prime\prime}<x^{\prime}$, then $u(x^{\prime})> u(x^{\prime\prime})$, since $u$ is increasing in $[0,x^{\star}]$.
The definition of $\succcurlyeq $ implies that $x^{\prime\prime}\not\succ x^{\prime}$, a contradiction.
If $x^{\prime\prime}>x^{\prime}$, then $\delta_{x^{\prime},x_{s}}<\delta_{x^{\prime\prime},x_{s}}$.
This implies $c_1(\delta_{x^{\prime},x_{s}})<c_1(\delta_{x^{\prime\prime},x_{s}})$.
The definition of $\succcurlyeq$ yields $x^{\prime\prime}\not\succ x^{\prime}$, a contradiction.
The argument is symmetric when $x_{s}>x^{\star}$.
\qed

\medskip
\noindent {\bf\textit{Proof of Proposition~\ref{PROP:Optimal choice}}}.  
	Since $U$ is a function upper semicontinuous in $x$, and due to Proposition \ref{PROP:One-many interval} the set $\max(\mathbb{R}_{+},\succcurlyeq)=[\min\{x_{s},x^{\star}\},\max\{x_{s},x^{\star}\}]$ is a (linearly ordered) closed interval (thus, it is compact), we can apply the  generalized Weierstrass' theorem \citep{MartinezLegaz2014,Quartieri2022} to conclude that $\argmax_{x\in\left(\mathbb{R}_{+},\succcurlyeq\right)}U\big(u(x),c_1(\delta_{x,x_s}),c_2(\delta_{x,\bar{x}_f})\big)$ exists.
\qed

\medskip
\noindent {\bf\textit{Proof of Proposition~\ref{PROP:two_sequential_rat}}}.
	By Proposition~\ref{PROP:Optimal choice} $\max_{x\, \in\, \max\left(\mathbb{R}_{+},\succcurlyeq\right)} U\big(u(x),c_1(\delta_{x,x_s}),c_{\,2}(\delta_{x,\bar{x}_f})\big)$ exists.
	Thus, $\Gamma^{\succcurlyeq}$ is a partial choice correspondence.
	We are left to show that there are two asymmetric relations $(\succ_1,\succ_2)$ on $\mathbb{R_+}$ such that $\Gamma^{\succcurlyeq}(\mathbb{R}_{+})=\max(\max(\mathbb{R_+},\succ_1),\succ_2)$.
	Assume that $\succ_1\equiv \succ$, and let $\succ_2\subset \mathbb{R_+}^2$ be defined by 
	
	$$x\succ_2 y \hspace{0.5cm}\;\text{if}\; U\big(u(x),c_1(\delta_{x,x_s}),c_{\,2}(\delta_{x,\bar{x}_f})\big) > U\big(u(y),c_1(\delta_{y,x_s}),c_{\,2}(\delta_{y,\bar{x}_f})).$$
Note that $\succ_1$ and $\succ_2$ are asymmetric.
Moreover,
\begin{align*}
&\Gamma^{\,\succcurlyeq}(\mathbb{R}_+)=\max_{x\, \in\, \max\left(\mathbb{R}_+,\succcurlyeq\right)} U\big(u(x),c_1(\delta_{x,x_s}),c_{\,2}(\delta_{x,\bar{x}_f})\big)=\max_{x\, \in\, \max\left(\mathbb{R}_+,\succ\right)} U\big(u(x),c_1(\delta_{x,x_s}),c_{\,2}(\delta_{x,\bar{x}_f})\big)=\\	
&\max(\max(\mathbb{R}_+,\succ),\succ_2)=\max(\max(\mathbb{R}_+,\succ_1),\succ_2).\hspace{8cm}
\qed \end{align*}

\medskip
{\noindent {\bf\textit{Proof of Proposition~\ref{PROP:standard_equilibrium}}}.
	Let $\mathscr{P}\subset\mathbb{R}_{+}^{n}$ be the set $\left\{(x_1,\ldots,x_n)\,\big\vert\, x_i\leq \dot{x} \,\text{for any} \,1\leq i\leq n\right\}$.
For each $(x_{1},\cdots,x_{n})\in\mathscr{P}$, define the map $\phi:\mathscr{P}\to2^{\mathscr{P}}$ as follows:
\vspace{-0.3cm}
$$\phi(x_{1},\cdots,x_{n})=\left\{(x^{\diamond}_{1},\cdots,x^{\diamond}_{n})\in\mathscr{P}\,\big\vert\, x^{\diamond}_{i}\in\argmax_{x_{i}}U_i(x_{i},x_{-i})\;\text{for all}\;1\leq i\leq n\right\}.$$
 Since $U_i$ is continuous for any $i\in\{1,\cdots,n\}$, Weierstrass's extreme value theorem yields that $\phi(x_{1},\cdots,x_{n})\neq \varnothing $ for any vector $(x_{1},\cdots,x_{n})\in \mathscr{P}$.
Continuity of $U_i$ also implies that $\phi(x_{1},\cdots,x_{n})$ is closed for any $(x_{1},\cdots,x_{n})\in \mathscr{P}$.
Since for each $i\in\{1,\cdots,n\}$ the function $U_i$ is quasiconcave in $x_{i}$, we obtain that $\phi(x_{1},\cdots,x_{n})$ is a convex set for any $(x_{1},\cdots,x_{n})\in\mathscr{P}$.

We are left to show that $\phi$ is upper hemicontinuous.
Denote by $\mathscr{C}(\mathscr{P})$ the family of all nonempty and closed subsets of $\mathscr{P}$.
Moreover, denote by $\mathbold{x}$ a vector $(x_{1},\dots,x_{n})\in\mathscr{P}$, and by ${\mathbold{x}}^{\diamond}$ a vector belonging to $\phi(\mathbold{x})$.  
 A set valued function $\phi\colon\mathscr{P}\to\mathscr{C}(\mathscr{P})$ is upper hemicontinuous if for any sequences $(\mathbold{x_m})_{m\in\mathbb{N}}$ and $({\mathbold{x}}^{\diamond}_m)_{m\in\mathbb{N}}$ such that ${\mathbold{x}}^{\diamond}_m\in\phi(\mathbold{x}_m)$ for any $m\in\mathbb{N}$, $\lim_{m\to\infty}(\mathbold{x}_m)=  \mathbold{x}^{*}$, and $\lim_{m\to\infty}({\mathbold{x}}^{\diamond}_m)= {{\mathbold{x}}^{\diamond}}^{*}$, we have that ${{\mathbold{x}}^{\diamond}}^{*} \in\phi(\mathbold{x}^{*})$ holds. 
%
Since $U_i$ is continuous, $\max_{x_i} U_i(x_i,x_{-i})$ is continuous, and we can apply Berge's maximum theorem to conclude that $\phi$ is upper hemicontinuous.
Since $\mathscr{P}$ is a nonempty, closed, bounded, and convex set, and $\phi$ is a nonempty, closed, and convex valued upper hemicontinuous function we can apply Kakutani's fixed point theorem to conclude that there is an equilibrium.
\qed}

\medskip

{\noindent {\bf\textit{Proof of Proposition~\ref{PROP:general_model}}}.
	Let $\mathscr{P}\subset\mathbb{R}_{+}^{n}$ be the set $\left\{(x_1,\ldots,x_n)\,\big\vert\, x_i\leq \dot{x} \,\text{for any} \,1\leq i\leq n\right\}$.
Moreover, denote by $\mathscr{P}_i$ the set $[0,\dot{x}]$.
For each $(x_{1},\cdots,x_{n})\in\mathscr{P}$, define the map $\phi:\mathscr{P}\to2^{\mathscr{P}}$ as follows:
\vspace{-0.3cm}
$$\phi(x_{1},\cdots,x_{n})=\left\{(x^{\diamond}_{1},\cdots,x^{\diamond}_{n})\in\mathscr{P}\,\big\vert\, x^{\diamond}_{i}\in\argmax_{x_{i}\in\max\left(\mathscr{P}_i,\succcurlyeq_{i,g(x_{-i})}\right)}U_i(x_{i},x_{-i})\;\text{for all}\;1 \leq i\leq n\right\}.$$

\noindent 
Since $U_i$ is continuous, $u_i$ is strictly quasiconcave, and $c_{1}^{\,i}$ is increasing in $\delta$ for any $i\in\{1,\dots,n\}$, and $\delta$ is the euclidean distance on $\mathbb{R}_{+}$, we apply  Proposition~\ref{PROP:Optimal choice} to obtain $\phi(x_{1},\cdots,x_{n})\neq\varnothing$ for any $(x_{1},\cdots,x_{n})\in \mathscr{P}$.
Continuity of $U_i$ implies that $\phi(x_{1},\cdots,x_{n})$ is closed for any $(x_{1},\cdots,x_{n})\in \mathscr{P}$.
Since for each $i\in\{1,\cdots,n\}$ the function $U_i$ is quasiconcave in $x_{i}$, we conclude that $\phi(x_{1},\cdots,x_{n})$ is a convex set for any $(x_{1},\cdots,x_{n})\in\mathscr{P}$.
We are left to show that $\phi$ is upper  
hemicontinuous.
Since $U_i$ is continuous, 
$\max_{x_i\in\max(\mathscr{P}_i,\succcurlyeq)} U_i(x_i,x_{-i})$ 
is continuous.
Thus, the Berge's maximum theorem yields that $\phi$ is upper hemicontinuous.
Since $\mathscr{P}$ is a nonempty, closed, bounded, and convex set, and $\phi$ is a nonempty, closed, and convex valued upper hemicontinuous function we apply Kakutani's fixed point theorem to conclude that there is an equilibrium after deferral.
\qed}

\medskip
\noindent {\bf\textit{Proof of Proposition~\ref{PROP:main proposition}}}.
	To show that in any equilibrium after deferral $x^{\diamond}_{i}=x^{\diamond}_{-i}\,$ holds, first note that, since $u_{1}(x)=u_{2}(x)$ for any $x\in\mathbb{R}_+$, we have that $x^{\star}_{1}=x^{\star}_{2}=x^{\star}$.
Moreover, Proposition~\ref{PROP:One-many interval} yields that $\max(\mathbb{R}_{+},\succcurlyeq _i)=[\min\{x_{-i},x^{\star}\},\max\{x_{-i},x^{\star}\}]$ for any $i\in\{1,2\}$.
Assume now toward a contradiction that $(x^{\diamond}_1,x^{\diamond}_2)$ is an equilibrium after deferral, and $x^{\diamond}_1\neq x^{\diamond}_2$.
Without loss of generality, assume that $x^{\diamond}_1 < x^{\diamond}_2$.
Three cases are possible: (i) $x^{\star} \leq x^{\diamond}_1 $, (ii) $x^{\diamond}_1 < x^{\star} \leq x^{\diamond}_2 $, and (iii) $x^{\diamond}_1 < x^{\diamond}_2 < x^{\star}$.
If (i) holds, the pair $(x^{\diamond}_1,x^{\diamond}_2)$ is not an equilibrium after deferral, because  $x^{\diamond}_2\not\in\max(\mathbb{R}_+,\succcurlyeq_{2,x^{\diamond}_1})=[x^{\star},x^{\diamond}_1]$.
If (ii) holds, $x^{\diamond}_1\not\in\max(\mathbb{R}_+,\succcurlyeq_{1,x^{\diamond}_2})=[x^{\star},x^{\diamond}_2]$, thus $(x^{\diamond}_1,x^{\diamond}_2)$ is not an equilibrium after deferral.
Finally, if (iii) is true, $x^{\diamond}_1\not\in\max(\mathbb{R}_+,\succcurlyeq_{1,x^{\diamond}_2})=[x^{\diamond}_2,x^{\star}]$, thus $(x^{\diamond}_1,x^{\diamond}_2)$ is not an equilibrium after deferral.
We conclude that $x^{\diamond}_i=x^{\diamond}_{-i}$.

Assume now that the pair $(x^{\diamond}_{i},x^{\diamond}_{-i})$ is an equilibrium such that $x^{\diamond}_{i}=x^{\diamond}_{-i}$.
For each $i\in\{1,2\}$ we have that $x^{\diamond}_i\in\max(\mathbb{R}_+,\succcurlyeq_{i,x^{\diamond}_{-i}})=[\min\{x^{\diamond}_{-i},x^{\star}\},\max\{x^{\diamond}_{-i},x^{\star}\}]$.
Moreover, Definition~\ref{DEF:Equilibrium_standard} yields $x^{\diamond}_{i}=\argmax_{x_i}U_{i}(x_i,x^{\diamond}_{-i})$ for each $i\in\{1,2\}$, which implies that $x^{\diamond}_{i}=\argmax_{x_i\in\max\left(\mathbb{R}_+,\succcurlyeq_{i,x^{\diamond}_{-i}}\right)}U_{i}(x_i,x^{\diamond}_{-i})$ for each $i\in\{1,2\}$.
By Definition~\ref{DEF:Equilibrium_deferral}, $(x^{\diamond}_{i},x^{\diamond}_{-i})$ is an equilibrium after deferral.
\qed

\setlength{\bibsep}{5pt}
\setstretch{0.97}

\setlength{\bibsep}{5pt}
\setstretch{0.97}


\end{document}